\documentclass[a4paper, 11pt]{article}
\usepackage[margin=1in]{geometry} 
\usepackage{xcolor}
\definecolor{kindablue}{RGB}{48,97,140}
\usepackage{enumitem}
\usepackage{blindtext}
\usepackage{graphicx, wrapfig}
\usepackage{subcaption}
\usepackage{algorithm, algpseudocode}
\usepackage{amssymb,amsmath}
\usepackage[authoryear]{natbib}
\usepackage{url}
\usepackage{comment}
\usepackage{bm}
\usepackage{amsthm,mathtools}
\usepackage{hyperref}
\usepackage{mathabx}
\usepackage{booktabs, makecell}

\newcommand{\mN}{\mathcal N} 
\newcommand{\mX}{\mathcal X} 
\newcommand{\mR}{\mathcal R} 
 
\newcommand{\mT}{\mathcal T} 
 
\newcommand{\mB}{\mathcal B} 
\newcommand{\mZ}{\mathcal Z} 
 
\newcommand{\mS}{\mathcal S}

\newcommand{\ba}{\boldsymbol a} 
\newcommand{\bb}{\boldsymbol b}

\newcommand{\bi}{\boldsymbol i} 
\newcommand{\bj}{\boldsymbol j} 
\newcommand{\bs}{\boldsymbol s}

\newcommand{\bx}{\boldsymbol x}

\newcommand{\bA}{\boldsymbol A} 
\newcommand{\bB}{\boldsymbol B} 
\newcommand{\bC}{\boldsymbol C} 
\newcommand{\bD}{\boldsymbol D} 
 
\newcommand{\bI}{\boldsymbol I} 
 
\newcommand{\bS}{\boldsymbol S} 
\newcommand{\bW}{\boldsymbol W} 
\newcommand{\bX}{\boldsymbol X} 
 
\newcommand{\bG}{\boldsymbol G} 
\newcommand{\bH}{\boldsymbol H}

\newcommand{\bZ}{\boldsymbol Z}

\newcommand{\bone}{\boldsymbol 1} 

\newcommand{\bPi}{\boldsymbol \Pi} 
\newcommand{\bOmega}{\boldsymbol \Omega} 
\newcommand{\blambda}{\boldsymbol \lambda} 
 
\newcommand{\bgamma}{\boldsymbol \gamma} 
 
\newtheorem{theorem}{Theorem}
\newtheorem{proposition}[theorem]{Proposition}

\newcommand{\vecc}{\text{vec} }
\newcommand{\diag}{\text{diag} }
 
\newcommand{\mbbR}{\mathbb R} 

\newcommand\sdots{\hbox to 1em{.\hss.\hss.}}

\DeclareMathOperator*{\argmax}{arg\,max}
\DeclareMathOperator*{\argmin}{arg\,min}

\title{Heteroskedastic Canonical Polyadic Tensor Decomposition}
\date{}
\author{
Kyle Ritscher$^{1,2}$, Carlos~Llosa-Vite$^1$
}
\date{%
    $^1$\textit{Computer Science Research Institute}, Sandia National Laboratories, Albuquerque, NM\\%
    $^2$\textit{Department of Statistics \& Data Science}, University of California, Los Angeles, CA
}
\begin{document}

\maketitle

\begin{abstract}
When minimizing the squared-error loss, the popular CP decomposition can be interpreted as parameter inference in a Gaussian model with a low-rank mean tensor and constant variance across the tensor entries. 
We introduce heteroskedastic-CP (HCP), which models entrywise variability with a non-constant, low-rank precision tensor, and develop an alternating block-coordinate ascent method to recover both the low-rank mean and precision tensors from noisy observations. Our procedure is computationally competitive, with the same leading-order factor-update complexity as CP-ALS. We demonstrate HCP on synthetic experiments and an EEG application.

\begin{figure}[b!]
    \centering
    \noindent
\begin{minipage}[c]{0.36\textwidth}
    \centering
    \includegraphics[width=\linewidth]{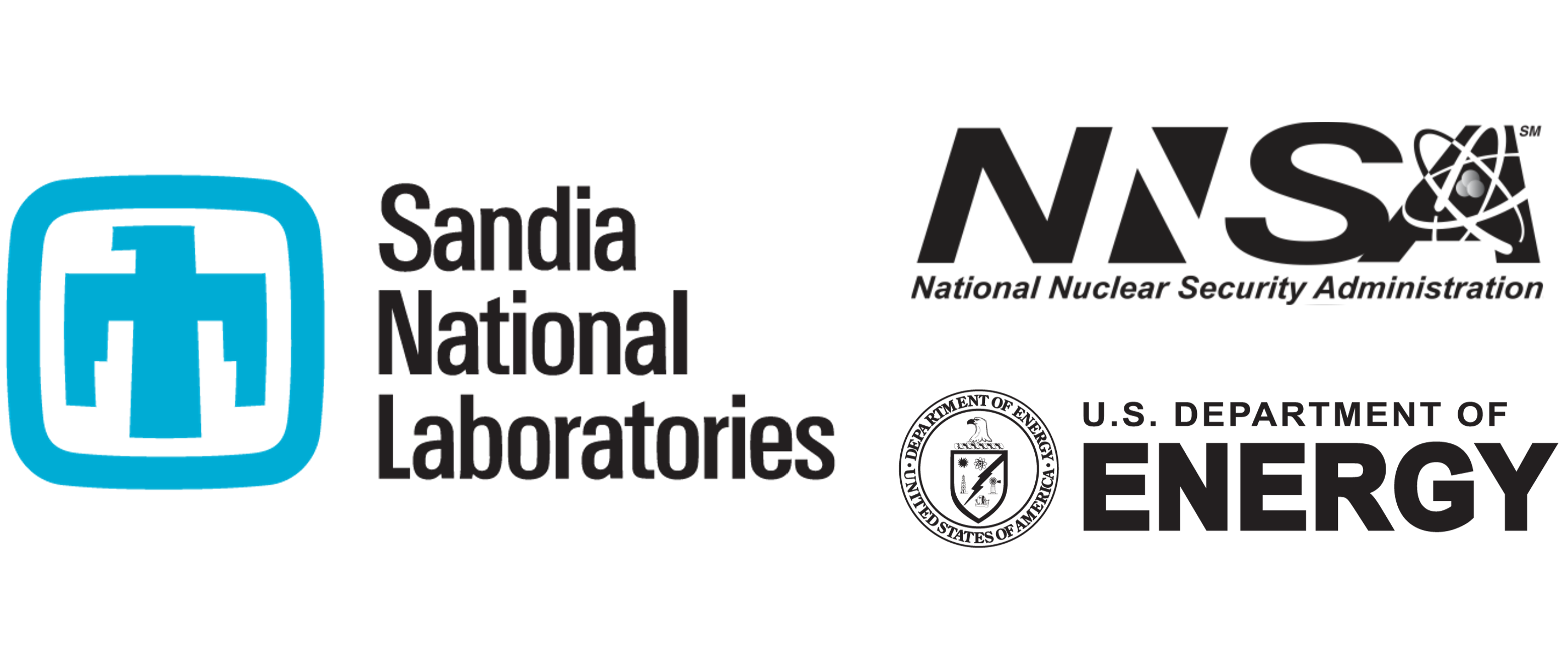}
\end{minipage}
\hfill
\begin{minipage}[c]{0.63\textwidth}
{\scriptsize \hyphenpenalty=10000 
This work was supported by the Laboratory Directed Research and Development program (Project 237043) at Sandia National Laboratories, a multimission laboratory managed and operated by National Technology and Engineering Solutions of Sandia LLC, a wholly owned subsidiary of Honeywell International Inc. for the U.S. Department of Energy’s National Nuclear Security Administration under contract DE-NA0003525. SAND2026-27105O
\par}
\end{minipage}

\end{figure}

\end{abstract}

\section{Introduction} \label{sec:intro}

The canonical polyadic (CP) decomposition, also known as the CANDECOMP/PARAFAC decomposition, is a fundamental tool in tensor analysis \citep{hitchcock1927, harshman1970parafac, carrollandchang70, koldaandbader09} that serves as a
higher-order generalization of the matrix singular value decomposition. Specifically, the CP decomposition approximates a tensor $\mathcal{X} \in \mathbb{R}^{N_1 \times N_2 \times \cdots \times N_D}$ as a sum of $R$ rank-one tensors:
\begin{equation} \label{eq:CPD}
    \mathcal{X} \approx \sum_{r=1}^{R} \lambda_r\mathbf{b}_r^{(1)} \circ \mathbf{b}_r^{(2)} \circ \cdots \circ \mathbf{b}_r^{(D)} = [\![\blambda;  \bB_1, \bB_2, \ldots, \bB_D ]\!].
\end{equation}
Here, $\circ$ denotes the outer product for $\mathbf{b}_r^{(d)}$ column vectors, and the equality is a notational switch to the Kruskal operator with factor matrices (see Section~\ref{sec:background}). The CP decomposition can be used for interpretable factorization, dimensionality 
reduction, imputation, and compression, and has found a wide range of applications \citep{murphy2013parafac, jaffe2014eemparafac, cong2015tensor, williams2018tensor, kolda2005higherorder, kolda2006tophits, opsahl2009clustering, chen2002blind, muti2005tensor, beylkin2009multivariate, acar2009unsupervised, koldaandbader09}. 

As pointed out by \cite{Hong_GCP_2020}, the $\ell_2$ optimization problem posed by the CP decomposition
\begin{equation} \label{eq:CPD-ls}
    \argmin_{\blambda\in \mathbb{R}^{R}, 
    \bB_1 \in \mathbb{R}^{N_{1} \times R}, \ldots, \bB_D \in \mathbb{R}^{N_{D} \times R}} \left\| \mathcal{X} - [\![ \blambda;\bB_1, \ldots, \bB_D ]\!] \right\|_F^2
\end{equation}

is exactly the maximum likelihood estimation (MLE) problem for $\mB \coloneqq [\![\blambda;  \bB_1, \bB_2, \ldots, \bB_D ]\!]$ under the assumption that $\mathcal{X}$ follows an entrywise independent Gaussian model with low-rank mean $\mB$, and constant variance $\sigma^2$, i.e. $\mX_{\bj} \sim \mN( \mB_{\bj} , \sigma^2),$ where $\bj = (j_1,\ldots, j_D)$ is a multi-index. However, in practice tensor data often exhibit heteroskedastic noise: some sensors are noisier than others, and measurements from certain experimental conditions may be less precise than others. Ignoring this heteroskedasticity can cause the recovered factors to be disproportionately influenced by high-variance regions of the data.

In this paper we propose the Heteroskedastic CP (HCP) Decomposition, which will both estimate the variance structure and adjust the mean decomposition accordingly. In this sense, HCP is a ``two-parameter'' extension of standard CP.
Maintaining the MLE framework, our HCP model assumes that our data tensor $\mX\in\mbbR^{N_1\times \ldots \times N_D}$ is normally distributed, entrywise independently, i.e. 
\begin{equation} \label{eq:generative-model}
    \mX_{\bj} \sim \mN( \mB_{\bj} , \mS_{\bj}^{-1}),
\end{equation}
where \emph{both} $\mB$ and $\mS$ are low-rank parameter tensors of the same dimensions as $\mX$. 

The log-likelihood (up to constants) of our proposed method is
\begin{equation}\label{eq:joint-likelihood}
    \ell(\mB, \mS \mid \mX) = \frac{1}{2}\big(\sum_{\bj} \log(\mS_{\bj}) - \sum_{\bj} \mS_{\bj}(\mX_{\bj} -\mB_{\bj})^2\Big).
\end{equation}

\subsection{Summary of Contributions} \label{sec:contributions}

In this article we propose an alternating block coordinate ascent procedure for optimizing Equation \eqref{eq:joint-likelihood} when both $\mS$ and $\mB$ are low-rank CP tensors. 
Despite the additional expressivity of HCP, this procedure is computationally efficient, with the cost per factor update dominated by the same term as in the popular CP-ALS (Alternating Least Squares) algorithm \citep{koldaandbader09}.
Our procedure relies on two novel subroutines that optimize Equation \eqref{eq:joint-likelihood} with either the precision tensor or the mean tensor held fixed. These are called Weighted CP and Gamma-rate CP, presented in Sections \ref{sec:mean} and \ref{sec:precision}, respectively. 

Since the log-likelihood of Equation~\eqref{eq:joint-likelihood} is unbounded above, the maximization problem can be degenerate in practice. We address this issue in Section~\ref{sec:heterocp} with a penalized-likelihood formulation that removes the degeneracy.

We demonstrate the improved estimation capabilities of HCP under heteroskedastic noise and showcase the unique ability of our ``two-parameter'' decomposition to adapt to unknown or misspecified generating distributions on synthetic experiments in Section~\ref{sec:experiments}. We also apply HCP to an EEG dataset in Section~\ref{sec:EEG}, where estimating the variance provides additional insights unavailable to a practitioner relying on CP alone.

\subsection{Notation \& Preliminaries} \label{sec:background}

In this section, we will provide a brief review of relevant tensor properties and discuss our notation. We provide a more detailed tensor review, define each matrix product, and provide some necessary prerequisite propositions in Appendix~\ref{sec:appendix-background}. 

We denote vectors (first-order tensors) by bold lowercase letters (e.g., $\ba$), matrices (second-order tensors) by bold uppercase letters (e.g., $\bA$), and higher-order (order three or greater) tensors by Euler script letters (e.g., $\mX$). The identity matrix of size $n \times n$ is $\bI_n$, and the length-$n$ vector of all ones is $\mathbf{1}_n$. The bracket notation $[n]$ denotes the set $\{1,\dots,n\}$. 
For a tensor $\mX \in \mathbb{R}^{N_1 \times \cdots \times N_D}$, the entry at multi-index $\bj = (j_1,\ldots,j_D)$, with $j_d \in [N_d]$, is written $\mX_{\bj}$. The following shorthands will be used throughout: $N_{-d} \coloneqq  \prod_{c \neq d} N_c$,  $N_{+} \coloneqq  \sum_{d=1}^D N_d$, and $N_\times \coloneqq  \prod_{d=1}^D N_d$. We define $\text{dv}(\mX) \coloneqq \diag(\vecc(\mX))$, where $\diag(\bx)$ is the diagonal matrix with $\bx$ as its diagonal, and $\vecc(\mX)$ denotes the vectorization. We write $\mX_{(d)} \in \mathbb{R}^{N_d \times N_{-d}}$ for mode-$d$ matricization, and $\|\mX\|_{F} = |\!|\vecc(\mX)|\!|_2$ denotes the Frobenius norm. 

We denote the Kronecker product of $\bA \in \mathbb{R}^{I\times J}$ and $\bB \in \mathbb{R}^{K\times L}$ as $\bA \otimes \bB \in \mathbb{R}^{IK\times JL}$, and the Khatri--Rao product (``matching columnwise'' Kronecker product) of $\bA \in \mathbb{R}^{I\times K}$ and $\bB \in \mathbb{R}^{J\times K}$ as $\bA \odot \bB \in \mathbb{R}^{IJ\times K}$. We denote the Hadamard (elementwise) product, division, and power by $\bA * \bB$, $\bA \oslash \bB$, and $\bA^{*b}$ respectively. We may also use each matrix product over a range of matrices, e.g. $ \odot_{d=1}^D \bA_d = \bA_1 \odot \dots \odot \bA_D$. 

The \emph{rank} of a tensor \citep{koldaandbader09} is the smallest integer $R$ such that $\mX$ can be written as a sum of $R$ rank-one tensors. Similarly, the non-negative rank \citep{limandcomon09} is the smallest integer $R$ such that $\mX$ can be written as a sum of $R$ rank-one tensors that have non-negative factors. Note however, that in this article we will only work with data tensors with unknown rank (such as $\mX$), and parameter tensors with known Kruskal form, such as
$\mB = [\![\blambda;  \bB_1, \bB_2, \ldots, \bB_D ]\!]$. Strictly speaking, if $\mB$ has a rank-$R$ Kruskal form then $\text{rank}(\mB) \leq R$. Kruskal tensors have a scaling ambiguity on the factor columns because $[\![\blambda*\bgamma;\bB_1,\bB_2,\dots, \bB_D]\!] = [\![ \blambda;\bB_1\diag(\bgamma),\bB_2,\dots, \bB_D]\!]$ holds for any vector $\bgamma$. Therefore, throughout this paper we will assume that the matrix factors of Kruskal tensors have unit column norms, while $\blambda \in \mathbb{R}_{>0}^{R}$ contains the scaling factor. Strictly speaking this normalization is not required but it avoids roundoff errors that can occur in extreme cases \citep{Ballard_Kolda_2025}. When a normalization step is required in our procedures, we use the function \textsc{ColumnNormalize($\widetilde\bB, p$)} which returns a tuple $(\bB,\blambda)$ such that $\bB$ has unit-$p$ norms and $\widetilde\bB = \bB\diag(\blambda)$ holds. We provide this function in Appendix~\ref{sec:appendix-background}.

\subsection{Related Work} \label{sec:related-work}

Leveraging the connection between maximum likelihood estimation and CP decomposition to motivate more sophisticated alternatives is not unique to our work. \cite{chiandkolda12} improves on previous non-negative tensor factorizations by explicitly modeling the sparse count tensor with the Poisson distribution rather than the Gaussian distribution, and generalized CP (GCP) decomposition \cite{Hong_GCP_2020} demonstrates the advantage of its loss function flexibility through connections to different MLE problems. 

To handle heteroskedastic noise, Weighted PARAFAC \citep{PAATERO1997223} can incorporate user-specified entrywise precisions, but it is not adaptive to unknown heteroskedasticity; its dynamic reweighting is designed for robustness to outliers, not for learning a variance structure.
Some recently proposed Tucker-parameterized models are more adaptive to unknown variance profiles. For example, the sub-Gaussian method 
in \cite{han2021optimalstatisticalcomputationalframework} replaces the SVD initialization with HeteroskedasticPCA initialization, allowing them to consider non-identically distributed entries,  
and \cite{agterberg2024statisticalinferencelowranktensors} improves HOOI (higher-order orthogonal iteration) with more robust diagonal deletion initialization.  To the best of our knowledge, however, HCP is the first tensor decomposition to treat the heteroskedastic precision tensor itself as an unknown low-rank parameter and estimate it jointly with the low-rank mean.

Heteroskedasticity has been much more extensively studied for matrix data \citep{hong2022optimallyweightedpcahighdimensional, Hong_PPCAT_2021, zhang2021heteroskedasticpcaalgorithmoptimality, Agterberg_2022, cai2020subspaceestimationunbalancedincomplete, Leeb_2021, zhang2024matrixdenoisingdoublyheteroscedastic, su2024datadrivenoptimalshrinkagesingular, behne2022fundamentallimitsrankonematrix}. We discuss our work's relation to these methods in Appendix~\ref{sec:appendix-HPCA}.

\section{Weighted CP Decomposition}
\label{sec:mean}

In this section, we optimize Equation~\eqref{eq:joint-likelihood} when the precision tensor $\mS$ is fixed. This subroutine amounts to solving a weighted CP decomposition problem.

Let $\mS \in \mathbb{R}^{N_1 \times \ldots \times N_D}_{\geq 0}$ be known ``weights'', and consider the objective function
\begin{equation} \label{eq:likelihood-mean} 
\ell(\mB \mid \mX, \mS)
=
-\dfrac{1}{2}
\sum_{\bj}
\mS_{\bj}(\mX_{\bj} - \mB_{\bj})^2. 
\end{equation} 
For entrywise positive $\mS$, Equation~\eqref{eq:likelihood-mean} corresponds to the likelihood for our generative model (Equation \eqref{eq:generative-model}) when the precision tensor $\mS$ is fixed and known. 
In the three-way case, this is the same data-fit term as in \citep{PAATERO1997223}, for which PMF3 computes Gauss--Newton updates to the three factor matrices simultaneously. In addition, for binary weights, Equation~\eqref{eq:likelihood-mean} corresponds to the masked least-squares objective used in CP tensor completion, solved with a joint gradient approach called CP-WOPT by \cite{Acar_2011_incomplete}. We explore the connection to tensor completion in Appendices~\ref{sec:model-selection} and~\ref{sec:appendix-TC}. 

We can express our optimization problem as a weighted least squares problem via the following steps:  
\begin{align}\label{eq:likelihood-mean2} 
-2\ell(\mB \mid \mX, \mS) &= 
-\dfrac{1}{2}
\sum_{\bj}
\mS_{\bj}(\mX_{\bj} - \mB_{\bj})^2 = 
\left\|
\mS^{*\frac{1}{2}}_{(d)}
*
(\mX_{(d)} - \widetilde\bB_d\bPi_{d})
\right\|_F^2
\\&\hfill=
\left\|
\operatorname{dv}\!\Big((\mS^{*\frac{1}{2}}_{(d)})^\top\Big)\Big(
\vecc(\mX_{(d)}^\top)\notag 
-
(\bI_{N_d} \otimes \bPi_{d}^\top)\vecc(\widetilde\bB_d^\top)
\Big)
\right\|_2^2
\notag
\end{align}

The second equality follows from Proposition~\ref{prop:kruskal-decomp}, isolating a single factor matrix in the typical ALS fashion, with $\widetilde\bB_d\coloneqq \bB_{d} \diag(\blambda_{\mB})$ and $\bPi_{d} \coloneqq (\odot_{c \neq d} \bB_c)^\top \in \mathbb{R}^{R_{\mB} \times N_{-d}} $. The third equality takes the transpose, vectorizes, and leverages Proposition~\ref{prop:vec-trick} to convert into a typical least squares form. The transpose is taken so that the Gram matrix becomes block diagonal. 

The weighted least-squares expression of Equation~\eqref{eq:likelihood-mean2} 
has optima characterized by the normal equations
\begin{equation}\label{eq:mean-update}
    \bG_d \vecc(\widetilde\bB_d^\top)
    =
    \vecc\!\left(
        \bPi_{d}(\mS * \mX)_{(d)}^\top
    \right),
\end{equation}
where $\bG_d$ is an $R_{\mB}N_d \times R_{\mB}N_d$ block-diagonal matrix whose 
$(i_d,i_d)$ block, for $i_d \in [N_d]$, is
$$
\bPi_{d}
    \diag\!\left((\mS_{(d)})_{i_d,:}\right)
    \bPi_{d}^\top
    \in \mathbb{R}^{R_{\mB}\times R_{\mB}}.
    $$
   
Because $\bG_d$ is block diagonal, the linear system separates across the rows of $\bB_d$ and can be solved in parallel.  
As in standard CP-ALS, one can solve these equations in multiple ways depending on numerical conditioning \citep{Ballard_Kolda_2025}. In our implementation, we use the Moore--Penrose pseudoinverse. 

To estimate $\mB \coloneqq [\![\blambda_{\mB};  \bB_1, \bB_2, \ldots, \bB_D ]\!]$ we sequentially apply Equation \eqref{eq:mean-update} to each $d = 1,2,\dots,D$ for $i_{\mB}$ iterations (or until convergence; see Appendix~\ref{sec:model-selection} for a discussion of convergence criteria). Algorithm~\ref{alg:heterocp-mean} displays the procedure, and we present a basic example in Appendix~\ref{sec:appendix-basic-ex-mean}.

\begin{algorithm}
\caption{Weighted CP Decomposition}
\label{alg:heterocp-mean}
\begin{algorithmic}[1]
    \State \textbf{Input:} data $\mathcal{X} \in \mathbb{R}^{N_1 \times \cdots \times N_D}$
    \State \textbf{Input:} weights $\mathcal{S} \in \mathbb{R}_{\geq 0}^{N_1 \times \cdots \times N_D}$ given as dense, sparse, or factored (faster computation of line \ref{algsec:Gcompute})
    \State \textbf{Input:} desired CP rank $R_{\mB}$
    \State \textbf{Input:} maximum number of iterations $i_{\mB}$
    \State \textbf{Initial Values:} $\mathbf{B}_2,\!\dots\!,\mathbf{B}_D$, unit column $\ell_2$ norms
\Function{WCP}{$\mathcal{X}, R_{\mathcal{B}}, \mathcal{S}, i_{\mB}$}
    \While{not converged and iterations $\leq i_{\mB}$}
        \For{$d = 1, \ldots, D$}
            \State $\bPi_{d} \gets (\odot_{c \neq d} \mathbf{B}_c)^\top$
                \State\label{algsec:mttkrp} $\bZ_d \gets
                (\mathcal{S} * \mathcal{X})_{(d)}\bPi_{d}^\top$
                \hfill (\texttt{MTTKRP})
            \For{$i_d = 1, \ldots, N_d$}
                \State\label{algsec:Gcompute} $\mathbf{G} \gets
                \bPi_{d}
                \operatorname{dv}\bigl((\mathcal{S}_{(d)})_{i_d,:}\bigr)
                \bPi_{d}^\top$
                \State $(\mathbf{B}_d)_{i_d,:} \gets \mathbf{b}^\top \text{where } \mathbf{b} \text{ solves } \mathbf{G}\mathbf{b} = (\bZ_d)_{i_d, :}^\top$
            \EndFor
            \State $(\mathbf{B}_d, \boldsymbol{\lambda}_{\mathcal{B}}) \gets \Call{\textsc{ColumnNormalize}}{\mathbf{B}_d, 2}$
        \EndFor
    \EndWhile
    \State \Return CP factors: $\{\boldsymbol{\lambda}_{\mathcal{B}}, \mathbf{B}_1,\dots,\mathbf{B}_D\}$
\EndFunction

\end{algorithmic}
\end{algorithm}

When $\mS$ is a dense tensor of weights, the complexity of Algorithm \ref{alg:heterocp-mean} is dominated by line \ref{algsec:Gcompute}, which has a cost of $O(R_{\mB}^2N_{\times})$ floating point operations for updating one factor matrix. If $\mS$ is sparse (as in tensor completion), then this complexity drops to $O(R_{\mB}^2\text{nnz}(\mS))$, where $\text{nnz}(.)$ denotes the number of non-zeros. This decrease from $N_{\times}$ to $\text{nnz}(\mS)$ can be of several orders of magnitude. 

Furthermore, if the rank-$R_{\mS}$ CP factorization $\mS = [\![ \blambda_{\mS};\bS_1, \bS_2, \ldots, \bS_D ]\!]$ holds (as is the case in our proposed HCP), then the Gram matrix $\bG_d$ has the closed form
\begin{equation} \label{eq:fast-gram}
    \bG_d = \sum_{r=1}^{R_{\mS}} (\blambda_{\mS})_r
    \operatorname{diag}\!\left((\bS_d)_{:,r}\right)
    \otimes \left(\underset{c\neq d}{*} \bH_{r, c} \right) ,
\end{equation}
where $\bH_{r, c}    \coloneqq    \bB_c^\top    \diag\big((\bS_c)_{:,r}\big)    \bB_c \in \mathbb{R}^{R_{\mB}\times R_{\mB}},$ 

for $r \in [R_{\mS}]$. 
Forming this matrix has a complexity of only $O\!\left(R^2_{\mB} N_+ R_{\mS}\right)$, which is considerably cheaper than $O(R_{\mB}^2N_{\times})$ or $O(R_{\mB}^2\text{nnz}(\mS))$ (see Appendix~\ref{sec:appendix-complexity} for complexity derivations). In this case, the computational reduction is so drastic that complexity is now dominated by the matricized tensor times Khatri--Rao product (MTTKRP) of line \ref{algsec:mttkrp}, which is the same operation that dominates other CP estimation algorithms such as CP-ALS, CP-APR or CP-WOPT \cite{koldaandbader09,chiandkolda12,Hong_GCP_2020}. 
So, in our implementation of Algorithm \ref{alg:heterocp-mean} we allow $\mS$ to be provided as a factored tensor.

\section{Gamma-Rate CP Decomposition}
\label{sec:precision}

Next, we consider the model of Equation \eqref{eq:generative-model} when the mean tensor $\mB$ is fixed, and derive an algorithm for estimating the positive factor matrices of $\mS$, $\bS_d$. 

The log-likelihood, up to constants, is
\begin{align}
\ell&(\mS \mid \mX, \mB)
=
\dfrac{1}{2}
\left(
\sum_{\bj} \log(\mS_{\bj})
-
\sum_{\bj}
\mS_{\bj}(\mX_{\bj} - \mB_{\bj})^2
\right)
\notag \\
&=
\dfrac{1}{2}
\mathbf{1}_{N_d}^\top
\left(
\log(\mS_{(d)})
-
\mS_{(d)} * (\mX - \mB)^{*2}_{(d)}
\right)
\mathbf{1}_{N_{-d}}
\notag \\
&=
\dfrac{1}{2}
\mathbf{1}_{N_d}^\top
\left(
\log(\widetilde\bS_d\bOmega_{d})
-
(\widetilde\bS_d\bOmega_{d}) * \mR_{(d)}
\right)
\mathbf{1}_{N_{-d}},
\label{eq:likelihood_precision}
\end{align}
where $\mR=(\mX - \mB)^{*2}$ is the squared residual, $\widetilde \bS_d \coloneqq \bS_d \text{diag}(\blambda_{\mS})$, and $\bOmega_{d} \coloneqq (\odot_{c \neq d} \bS_c)^\top \in \mathbb{R}_{>0}^{R_{\mS} \times N_{-d}}$. The final equality follows from Proposition~\ref{prop:kruskal-decomp}. 

Because 
$\mS_{\bj}    \mR_{\bj}    \sim \chi^2_1 \equiv  \operatorname{Gamma}(\frac{1}{2}, \frac{1}{2}),$ Equation~\eqref{eq:likelihood_precision} is the log-likelihood under 
\begin{equation} \label{eq:gamma-connection}
    \mR_{\bj} \sim 
    \operatorname{Gamma}\Big(
        \frac{1}{2}, \text{rate} =
        \frac{1}{2}[\![ \blambda_{\mS};\bS_1, \ldots, \bS_D ]\!]_{\bj}
    \Big).
\end{equation}
Hence, we refer to this subproblem as Gamma-rate CP decomposition, since we parametrize the rate (equivalently, the inverse mean), rather than the mean as in Section 3.3 of \cite{Hong_GCP_2020}.

Inspired by NNMF, PTF, and NNTF \cite{leeandseung00, WELLING20011255, chiandkolda12}, we develop a minorize-maximize update rule for \(\widetilde\bS_d\). While a general-purpose constrained optimizer such as GCP could solve our problem, our result guarantees monotonic non-decrease without requiring a step size, and enforces non-negativity by construction, rather than relying on box constraints. We defer the proof to Appendix~\ref{sec:appendix-proofs}. 

\begin{theorem} \label{thm:GPCP} Let $\bA, \bC \in \mathbb{R}_{\geq 0}^{m \times n}$, $\bB\in \mathbb{R}_{> 0}^{m \times p}$, and $\bOmega \in \mathbb{R}_{> 0}^{p \times n}$. Assume that each row of \(\bA\) has at least one positive entry and \(\bC\bOmega^\top\) is entrywise positive. Then the function 
$
f(\bB) = \bone_m^\top\left\{\bA*\log(\bB\bOmega) - \bC*(\bB\bOmega)
\right\}\bone_n
$
is non-decreasing under the multiplicative update
\begin{equation} \label{eq:general-update}
    \bB\leftarrow \bB* \left\{ \left [ \bA\oslash (\bB\bOmega)\right]\bOmega^\top\right\} \oslash (\bC\bOmega^\top).
\end{equation}
\end{theorem}

For Equation \eqref{eq:likelihood_precision}, the update rule translates\footnote{$\bC\bOmega^\top = \mR_{(d)}(\bOmega_{d})^\top$ will have an entry of $0$ if an entire length $N_{-d}$ row of the matricized squared residual tensor, $\mR_{(d)}$, is exactly $0$. Such exact row-wise interpolation is unlikely when $\mX$ is noisy, and the constraint discussed in Section~\ref{sec:constraint} prevents this altogether.} 
(set $\bA$ to $\bone_{N_d} \bone_{N_{-d}}^\top$, 
$\bB$ to \(\widetilde{\bS}_d\), 
$\bOmega$ as $\bOmega_{d}$ and 
$\bC$ as $\mR_{(d)}$) to
\begin{equation}\label{eq:MM-update}
     \widetilde\bS_d \gets \widetilde\bS_d * \left(  (\widetilde\bS_d \bOmega_{d})^{*-1} \bOmega_{d}^\top\right)  \oslash\left( \mR_{(d)} \bOmega_{d}^\top\right). 
\end{equation}
We leave Theorem~\ref{thm:GPCP} in its general form due to its applicability beyond Gamma-rate CP. For example, the CP-APR subproblem and NNMF \citep{chiandkolda12, leeandseung00} are special cases where $\bA$ is the data and $\bC = \bone_{N_d} \bone_{N_{-d}}^\top$, and we can use Theorem~\ref{thm:GPCP} for tensor completion by letting $\bA$ be an observation mask (see Appendix~\ref{sec:appendix-TC}).

To estimate $\mS \coloneqq [\![\blambda_{\mS};  \bS_1, \bS_2, \ldots, \bS_D ]\!]$, we sequentially apply Equation \eqref{eq:MM-update} to each $d = 1,2,\dots,D$ for $i_{\mS}$ iterations or until convergence. Algorithm~\ref{alg:heterocp-precision} displays the procedure, and we present a basic example in Appendix~\ref{sec:appendix-basic-ex-var}. Note, in this algorithm we normalize with $\ell_1$ rather than $\ell_2$ norm.

\begin{algorithm}[H]
\caption{Gamma-rate CP Decomposition} \label{alg:heterocp-precision}
\begin{algorithmic} [1]
\State \textbf{Input:} Data: \( \mR \in \mathbb{R}_{\geq 0}^{N_1 \times N_2 \times \cdots \times N_D} \)
\State \textbf{Input:} Desired rank of results: $R_{\mS}$
\State \textbf{Input:} Maximum inner iterations: $i_{\text{in}}$
\State \textbf{Input:} Maximum outer iterations $i_{\mS}$
\State \textbf{Initial values:} $\blambda_{\mS}\in\mbbR^{R_{\mS}}_{>0}$, entrywise positive $(\bS_1,\!\dots\!,\bS_D)$, with unit column $\ell_1$ norms 
\Function{GammaCP}{$\mR, R_{\mS}, i_{\mS}, i_{\text{in}}$}
    \While{not converged and iterations $\leq i_{\mS}$}
        \For{$d = 1, \ldots, D$}
            \State $\bOmega_{d} \gets (\odot_{c \neq d} \bS_c)^\top$
            \State $\bD_d \gets \mR_{(d)} \bOmega_{d}^\top$
            \State \( \tilde \bS_d \gets \bS_d \cdot \text{diag}(\blambda_{\mS}) \) 
            \For{$i = 1, \ldots, i_{\text{in}}$}
                \State $\tilde \bS_d \gets \tilde \bS_d * \left(  (\tilde  \bS_d \bOmega_{d})^{*-1} \bOmega_{d}^\top\right)  \oslash \bD_d$
            \EndFor
            \State $\{\bS_d,\blambda_{\mS}\}\gets \textsc{ColumnNormalize}(\tilde\bS_d,1)$
        \EndFor
    \EndWhile
    \State \Return CP Factors: $\{\blambda_{\mS},\bS_1,\dots,\bS_D\}$
\EndFunction
\end{algorithmic}
\end{algorithm}

\section{Heteroskedastic CP Decomposition}
\label{sec:heterocp}
Now, we provide the full procedure in Algorithm~\ref{alg:heterocp} for jointly estimating a low-rank mean tensor, $\mB$, and a low-rank precision tensor $\mS$, from an observed noisy tensor $\mX$.
As discussed in Section~\ref{sec:intro}, the idea is to alternate between estimating $\mB$ with Algorithm~\ref{alg:heterocp-mean}, assuming $\mS$ is fixed and known, and estimating $\mS$ with Algorithm~\ref{alg:heterocp-precision}, assuming $\mB$ is fixed and known. 

For initialization, we use standard CP decomposition to initialize $\bB_d$, and Algorithm~\ref{alg:heterocp-precision} to initialize $\bS_d$. In our experiments (Sections~\ref{sec:experiments},~\ref{sec:EEG}), ranks will be either known or tailored to the application. Appendix~\ref{sec:model-selection} discusses data-driven rank selection and other settings (e.g. convergence criteria, maximum iterations) for Algorithms~\ref{alg:heterocp-mean}, \ref{alg:heterocp-precision}, and \ref{alg:heterocp}.
 \begin{algorithm} 
\caption{Heteroskedastic CP Decomposition}
\label{alg:heterocp}
\begin{algorithmic}[1] 
\State \textbf{Input:} Data \( \mX \in \mathbb{R}^{N_1 \times N_2 \times \cdots \times N_D} \) 
\State \textbf{Input:} Desired ranks of results: $R_{\mB}, R_{\mS} > 0$
\State \textbf{Input:} Submethod iterations: $i_{\mB}, i_{\mS}, i_{\text{in}} > 0$
\State \textbf{Input:} Maximum iterations: $i_{\text{tot}} > 0$
\State \textbf{Initial values:} $(\bB_2,\!\dots\!,\bB_D)$ as in Alg. \ref{alg:heterocp-mean}
\State \textbf{Initial values:} $(\blambda_{\mS},\bS_1,\!\dots\!,\bS_D)$ as in Alg. \ref{alg:heterocp-precision}

\While{not converged and iterations $\leq i_{\text{tot}}$}
    \State\label{alg:mean_subrout} $\{\boldsymbol{\lambda}_{\mathbf{B}}, \mathbf{B}_1,\dots,\mathbf{B}_D\}\gets \Call{\textsc{WCP}}{\mathcal{X}, R_{\mB}, \mS, i_{\mB}}$
    \State\label{alg:residualcomputation}  $\mR \gets (\mX  - [\![ \boldsymbol{\lambda}_{\mathbf{B}}; \mathbf{B}_1,\dots,\mathbf{B}_D]\!])^{*2}$
    \State\label{alg:var_subrout} $\{\boldsymbol{\lambda}_{\mathbf{S}}, \mathbf{S}_1,\dots,\mathbf{S}_D\}\gets \Call{\textsc{GammaCP}}{\mR, R_{\mS}, i_{\mS}, i_{\text{in}}}$
\EndWhile
\State \textbf{Output:} \( \bB_d \in \mathbb{R}^{N_d \times R_{\mB}} \) for \( d = 1, \ldots, D \)
\State \textbf{Output:} \( \bS_d \in \mathbb{R}_{>0}^{N_d \times R_{\mS}} \) for \( d = 1, \ldots, D \)
\end{algorithmic}
\end{algorithm}

The time complexity of Algorithm \ref{alg:heterocp} is
\begin{equation} \label{eq:complexity}
    O\left(
i_{\text{tot}}
\left[
i_{\mB}D N_\times R_{\mB}
+
i_{\mS} i_{\text{in}} D N_\times R_{\mS}
\right]
\right),
\end{equation}
where the first and second terms correspond to Algorithms \ref{alg:heterocp-mean} and \ref{alg:heterocp-precision}, respectively.
For reference, a sweep of CP-ALS has a cost of $O(DN_{\times}R_{\mB})$ which is the same cost as our Algorithm \ref{alg:heterocp-mean}, and, in most practical applications, Algorithm~\ref{alg:heterocp-precision}, as well, 
giving HCP similar performance to CP-ALS.
Deeper discussion and full derivations are deferred to Appendix~\ref{sec:appendix-complexity}. 

In practice, however, Algorithm~\ref{alg:heterocp} will often run into the degeneracy in the maximization problem. To see the mechanism, consider a single index $\bj$ where an estimate $\widehat {\mB_{\bj}}$ drives the residual $(\mX_{\bj} - \widehat {\mB_{\bj}})$ to zero; then, the corresponding estimate $\widehat {\mS_{\bj}}$, and the likelihood alongside it, will increase to $\infty$, unless the low rank coupling of entries forces enough residual cost elsewhere. This degeneracy is analogous to finding the variance of a single number, and occurs, for instance, in Gaussian mixture models when a cluster is selected to contain a single point. The remainder of this section describes modifications to Algorithm~\ref{alg:heterocp} that address this degeneracy. In Appendix~\ref{sec:appendix-heterocp} we give an expanded implementation, Algorithm~\ref{alg:heterocp-full}, that includes these modifications. 

\subsection{Rank-One Precision Tensor} \label{sec:rank-one} 

When the precision rank is one, the MM update (Equation \eqref{eq:MM-update}) reduces to
$$
\bs_d \gets N_{-d}\left(\mR_{(d)} \bOmega_{d}\right)^{*-1}, 
$$
where $\bs_d \in \mathbb{R}^{N_{d}}$  and $\bOmega_{d} = \bs_D \odot \dots \odot \bs_{d+1} \odot \bs_{d-1} \odot \dots \odot \bs_1 \in \mathbb{R}^{N_{-d}}$ are now \textit{vectors}. For fixed other factors, or along sequences where \(\bOmega_d\) is bounded below, an entry of this weighted residual contraction only tends to zero when the entire corresponding residual row tends to zero. Since this is unlikely whenever $\mX$ is noisy, we have found that setting $R_{\mS} =1$ tends to lead to a stable optimization problem in practical settings. 

\subsection{Partial Heteroskedasticity}
\label{sec:partial-hetero}
Our second adjustment is  ``partial heteroskedasticity'', i.e., allowing the precision tensor $\mS$ to vary along certain modes but not others. This assumption is common in applications where certain modes correspond to the dimension of a process, and other modes to samples of the process, as in MANOVA \citep{johnson07}, factor analysis  \citep{AndersonRubin1956}, or heteroskedastic PCA \cite{hong2022optimallyweightedpcahighdimensional, Hong_PPCAT_2021, zhang2021heteroskedasticpcaalgorithmoptimality}.

To simplify notation (we can always permute the tensor modes), we will assume the first $P \leq D$ modes of $\mS$ are heteroskedastic, while the final $Q = D - P$ modes are homoskedastic. 
This means that $\mB = [\![ \blambda_{\mB};\bB_1, \ldots, \bB_{D}]\!]$ as usual, while $\mS$ is written, 
\begin{equation}\label{eq:S_repeated}
    \mS = [\![\blambda_{\mS}; \bS_1, \ldots, \bS_P]\!] \circ \bone_{N_{P+1}} \circ \ldots \circ \bone_{N_{D}} =
    [\![\blambda_{\mS}; \bS_1, \ldots, \bS_P, \bone_{N_{P+1}}\bone_{R_{\mS}}^\top, \ldots, \bone_{N_{D}}\bone_{R_{\mS}}^\top]\!].
\end{equation}

Implementing this restriction is simple. In Algorithm \ref{alg:heterocp}, we replace $\mR$ in line \ref{alg:residualcomputation} with its average over the homoskedastic modes: for $\bi \in [N_{1}] \times \ldots \times [N_{P}]$
\begin{equation} \label{eq:avg-res}
    \mR_{\bi} = \frac{1}{\prod_{q = 1}^Q N_{P+q}} \sum_{\bj \in [N_{P+1}] \times \ldots \times [N_D]} (\mX_{(\bi, \bj)} - \mB_{(\bi, \bj)})^{2}. 
\end{equation}
The only other change is for initializing the subroutines; line \ref{alg:mean_subrout} now uses the second equality of Equation \eqref{eq:S_repeated} for $\mS$, and line \ref{alg:var_subrout} only initializes the first $P$ modes.

When leveraging this restriction, the method is far less likely to suffer from the degeneracy discussed; even with $Q = 1$, $\ell(\mB, \mS \mid\mX) \to \infty$ would require small residuals for an entire fiber. Observing the unbounded likelihood is still possible, though, particularly when $N_{P+1}, \ldots, N_{D}$ are small and $\mS$ has high rank.

\subsection{Constrained Residual Tensor} \label{sec:constraint}

We can avoid $\ell(\mS \mid\mX, \mB) \to \infty$ (and prevent divisions by zero) altogether by replacing $\mR$ with $\mR + \delta$ in the $\mS$ update, for the scalar $\delta > 0$. This approach changes the log-likelihood from Equation \eqref{eq:likelihood_precision} to
\begin{equation} \label{eq:penalized-likelihood}
\ell_{\delta}(\mS \mid \mX,\mB)
\coloneqq \dfrac{1}{2}\,
\mathbf{1}_{N_d}^{\top}
\Big[
    \log(\widetilde\bS_d\bOmega_{d})  -
    (\mR_{(d)}\!+\!\delta) * \widetilde\bS_d\bOmega_{d}
\Big]
\mathbf{1}_{N_{-d}} .
\end{equation}
The following theorem justifies optimizing $\ell_{\delta}$ as penalized estimation of $\ell$.

\begin{theorem}\label{thm:penalty}
Consider the scalar $\delta > 0$. The functions $\ell_{\delta}(\mS| \mX, \mB)$ (Equation \eqref{eq:penalized-likelihood}) and $\ell(\mS| \mX, \mB)$ (Equation \eqref{eq:likelihood_precision}) are related as 
$$\ell_{\delta}(\mS| \mX, \mB) =  l(\mS| \mX, \mB) - \frac{\delta}{2} \| \mS\|_{1}.$$
\end{theorem}

We defer the proof to the appendix, and note that Theorem~\ref{thm:penalty} holds for $\ell(\mB, \mS|\mX)$ and $\ell_{\delta}(\mB, \mS|\mX)$ as well. Implementation is also straightforward; we simply add $\delta$ to $\mR$ in line \ref{alg:residualcomputation} of Algorithm \ref{alg:heterocp}. We still inherit the non-decreasing result in Theorem~\ref{thm:GPCP}, as the residual modification will not violate the assumptions of the theorem. The $\mB$ update still solves $\ell(\mB |\mX, \mS)$ since the penalty only depends on $\mS$.

Examining $\ell(\mB |\mX, \mS)$ in the context of Theorem~\ref{thm:penalty} allows us to simply see that the maximization problem avoids the degeneracy because it is now bounded: 
\begin{equation}
\begin{aligned}
        \ell(\mB |\mX, \mS) &= \dfrac{1}{2}
\left(
\sum_{\bj} \log(\mS_{\bj})
-
\sum_{\bj}
\mS_{\bj}\Big((\mX_{\bj} - \mB_{\bj})^2 + \delta\Big) \right) \\
&\leq \dfrac{1}{2} \left(
\sum_{\bj} \log(\mS_{\bj})
-
\sum_{\bj}
\mS_{\bj}\delta \right) \leq -\frac{N_{\times}}{2}(1+\log \delta)
\end{aligned}
\end{equation}
The first inequality follows from $(\mX_{\bj} - \mB_{\bj})^2 \geq 0$, and the second, from the scalar function $\log s - \delta s$ being maximized at $s = \delta^{-1}$. This bound, combined with the inherited non-decreasing factor updates from Theorem~\ref{thm:GPCP} (and assuming mean factor updates in solve the linear system exactly, so the block maximization in Algorithm~\ref{alg:heterocp-mean} is also non-decreasing), shows that our procedure applied to $\ell(\mB |\mX, \mS)$ converges to a finite limit. 

Unless specified otherwise, we set $\delta$ to be $\sqrt{\text{Machine-Precision}} \coloneqq \sqrt{\epsilon_{\mathrm{mach}}}$.  

\section{Experiments} \label{sec:experiments}

In this section we perform synthetic experiments to demonstrate the advantages of HCP. For clarity, in Appendix~\ref{sec:appendix-experiment-table}, we summarize the experimental design choices across all experiments.

\subsection{Rotating Butterfly} \label{sec:butterfly}

\begin{wrapfigure}{r}{0.4\textwidth} 
    \vspace{-2em}
    \centering
    \includegraphics[width=\linewidth]{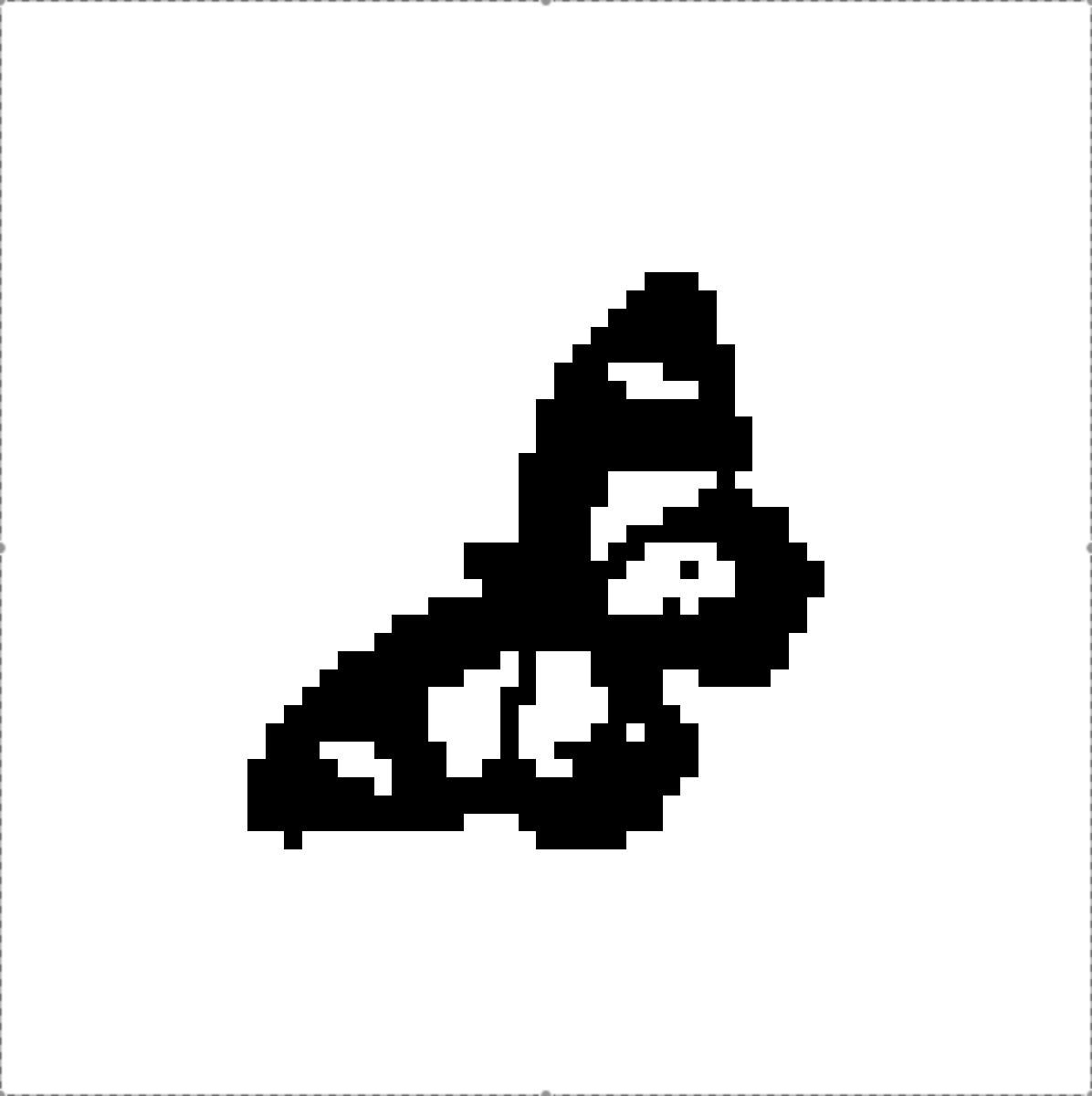} 
    \caption{$200 \times 200$ pixelated butterfly as the first frontal slice of $\mB$. }
    \label{fig:frame1}
\end{wrapfigure}
Our first experiment demonstrates the improvement of HCP over existing methods in the presence of heteroskedastic noise. We model a $200 \times 200$ binary matrix representing a pixelated butterfly image taken from MPEG-7 and available in the \texttt{fdasrvf} R package \citep{tucker2026fdasrvf}, displayed in Figure~\ref{fig:frame1}, where the third mode contains approximately \(\pi/8\)-radian pixelized rotations of the image over a full rotation (so $\mB \in \mathbb{R}^{200 \times 200 \times 17}$).

Then, representing the common situation where the noise is largest where the signal is, we generate our precision tensor to have small entries in the center, and constant along the third mode (which corresponds to rotation). The precision frontal slices have maximum value of two in the corners, and decay exponentially toward the center with a rate $\lambda$, which  ``controls'' the level of heteroskedasticity ($\lambda = 0$ represents constant entrywise variance). The precision tensor $\mS$ is then standardized to have an average variance of $0.5$, ensuring a constant SNR between values of $\lambda$. 
Figure~\ref{fig:Sgen} demonstrates a slice of $\mS$  with $\lambda \in \{1, 2, 4\}$. 

\begin{figure}[ht]
    \centering
    \includegraphics[width=0.8\textwidth]{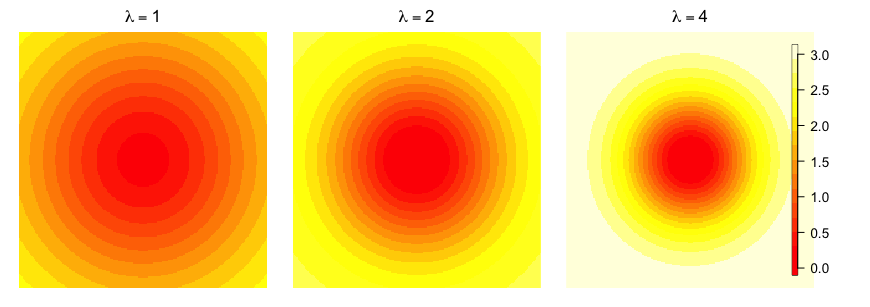} 
    \caption{Frontal slice of $\mS$ following heteroskedastic variance construction, for different values of $\lambda$.}
    \label{fig:Sgen}
\end{figure}

We generate $\mX_{\bj} \sim \mN(\mB_{\bj}, \mS_{\bj}^{-1})$, where $\mS$ has heteroskedasticity levels given by $\lambda \in \{0, 0.5, 1, 2, 3, 4, 6\}$, and fit HCP and CP. We chose mean rank $R_{\mB} = R = 16$ such that CP's mean recovery in the $\lambda = 0$ (constant variance) case was strong, and set $R_{\mS} = 3$. Figure~\ref{fig:mse-comp} displays the mean relative error over $10$ simulations in terms of Frobenius norm, $\frac{\| \widehat{\mB} - \mB\|_F}{\|\mB\|_F}$, and we demonstrate the visual recovery for a single slice in Figure~\ref{fig:slice-ex}. As we can see, when there is no heteroskedasticity ($\lambda = 0$) or low heteroskedasticity, the methods perform roughly the same, but as $\lambda$ grows CP degrades significantly, while HCP does not.

\begin{figure}[H]
    \centering
    \includegraphics[width=0.64\textwidth]{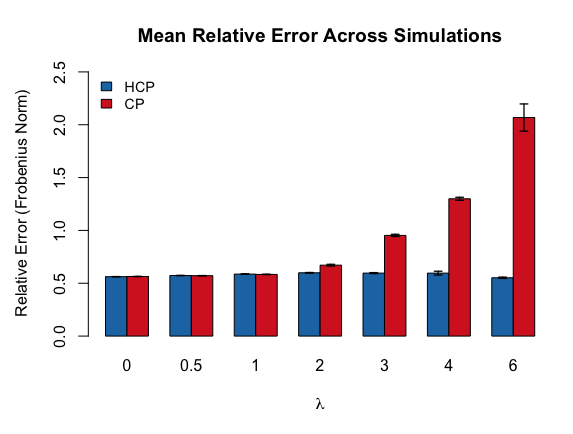}
    \caption{Comparison of relative errors in estimates of \(\mB\) at varying heteroskedasticity levels.}
    \label{fig:mse-comp}
\end{figure}

\begin{figure}[H]
    \centering
    \includegraphics[width=0.8\textwidth]{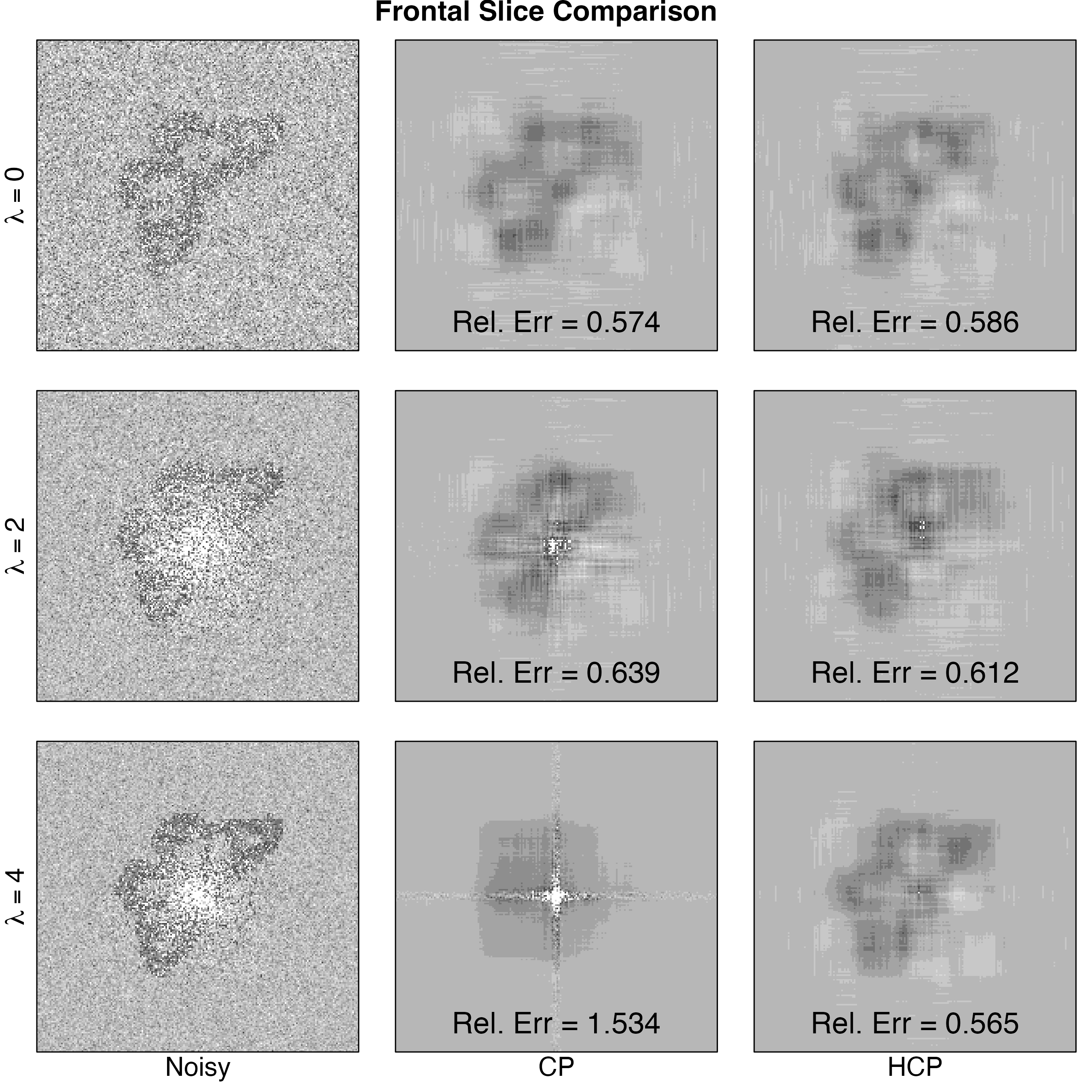}     \caption{Visualization of $9$th slice of the noisy observed tensor $\mX$ and the CP and HCP estimates $\widehat{\mB}$, for varying heteroskedasticity levels.}
    \label{fig:slice-ex}
\end{figure}

\subsection{Misspecified Generating Distribution} \label{sec:dist-study}
For our second experiment, we explore another advantage of HCP: by estimating both the conditional mean tensor and a second moment or variance-related tensor, HCP can be more safely employed when the exact noise model is not known.

We construct a simulation in which the true mean tensor \( \mB = [\![\blambda_{\mB};\bB_1,\bB_2,\bB_3]\!] \in \mathbb{R}^{40 \times 32 \times 48} \) has low CP rank \(R=3\) with $1/R$ of the entries in each column of $B_1,B_2,B_3$ generated from $\operatorname{Unif}(10,20)$, and the rest from $\operatorname{Unif}(1,2)$. The noisy observations are then generated from four distributions with different forms of heteroskedasticity:
\[
X_{\bj} \mid \mu_{\bj} \sim
\begin{cases}
\text{Gaussian}(\mu_{\bj}, \sigma^2) 
    & \operatorname{Var}(X_{\bj} \mid \mu_{\bj}) \, \propto \, 1 \\
\text{Poisson}(\mu_{\bj}) 
    & \operatorname{Var}(X_{\bj} \mid \mu_{\bj}) \, \propto \, \mu_{\bj} \\
\text{Gamma}(\alpha, \text{rate} = \frac{\alpha}{\mu_{\bj}}) 
    & \operatorname{Var}(X_{\bj} \mid \mu_{\bj}) \, \propto \, \mu_{\bj}^2 \\
\text{Inv-Gaussian}(\mu_{\bj}, \lambda)
    & \operatorname{Var}(X_{\bj} \mid \mu_{\bj}) \, \propto \, \mu_{\bj}^3 
\end{cases},
\]
where \(\mu_{\bj}\) denotes the corresponding entry of the true mean tensor \( \mathcal{\bB} \).
We calibrated the free dispersion parameters $\sigma^2,\alpha, \text{ and }\lambda$ so that the average conditional variance was the same across data-generating mechanisms. Using the Poisson model, which has no additional dispersion parameter, as the reference, our target average variance is $\frac{1}{n}\sum_i \mu_i$, so Gaussian standard deviation $\sigma$, Gamma shape parameter $\alpha$, and inverse Gaussian shape parameter $\lambda$ are set at 
\[
\sigma^2 = \frac{1}{n}\sum_i \mu_i, \qquad
\alpha = \frac{\frac{1}{n}\sum_i \mu_i^2}{\frac{1}{n}\sum_i \mu_i}, \qquad 
\lambda = \frac{\frac{1}{n}\sum_i \mu_i^3}{\frac{1}{n}\sum_i \mu_i}. 
\]

Figure~\ref{fig:dist-experiment} summarizes the relative errors in $\widehat{\mB}$ over $50$ simulations, comparing the mean estimates from HCP, CP, and GCP decompositions \citep{Hong_GCP_2020} using Poisson, Gamma, and Inv-Gaussian losses (GCP implementation details are in Appendix~\ref{sec:appendix-gcp}). We use ranks $R = R_{\mB} = R_{\mS} = 3$. As expected, the best fit comes from the correct distributional loss, but these methods can each struggle under different distributions. In contrast, HCP performs strongly in all cases, because even though our method has an inherent Gaussian assumption, estimating two moments allows it to be more flexible and adapt to model misspecification. 

\begin{figure}[]
    \centering
    \includegraphics[width=0.8\textwidth]{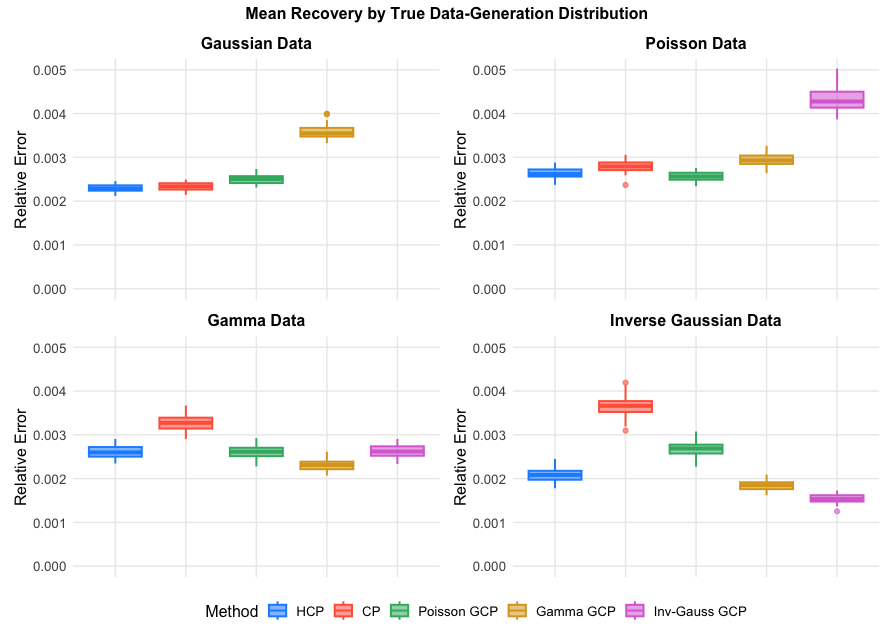}
    \caption{
    Panels correspond to a different true observation models, and boxplots show the $\widehat{\mB}$ recovery in relative Frobenius error over replicates for HCP, CP, and GCP methods using Poisson, Gamma, and inverse Gaussian losses.
    }
    \label{fig:dist-experiment}
\end{figure}

\subsection{Parameter Recovery Benchmark} \label{sec:simulation}

Finally, we conduct a simulation study across varying tensor sizes, $m_{\text{mult}} \cdot \left(10 \times 8 \times 12 \right)$, with $m_{\text{mult}} \in \{1.5, 2, 3, 6\}$, where we know the true low-rank tensors $\mB$ and $\mS$. We generate a partially heteroskedastic $\mS$ that is homoskedastic along the third mode, and we estimate it by leveraging Section~\ref{sec:partial-hetero} accordingly. To analyze the effect of $\delta$ in Section \ref{sec:constraint}, we use both $\delta = 0.001$ and $\delta = \sqrt{\epsilon_{\mathrm{mach}}}$, labeled HCP and HCP0, respectively. In Appendix~\ref{sec:appendix-simulation} we provide more details on the simulation, present results for fully heteroskedastic precision, and analyze the effect of rank misspecification. We measure performance with relative errors, and for the mean recovery we will compare our method against CP decomposition. For the precision, we are not aware of a directly comparable baseline.

Figure~\ref{fig:sim-results-low} provides our results as boxplots over 50 simulations. The mean recovery consistently outperforms CP decomposition, and there is almost no difference in the mean or precision recovery between HCP and HCP0, particularly for the larger tensor sizes. The sufficiency of a small constraint motivates our choice of $\delta = \sqrt{\epsilon_{\mathrm{mach}}}$ as the default. 

\begin{figure}[]
    \centering
    \begin{subfigure}[t]{0.49\linewidth}
        \centering
        \includegraphics[width=\linewidth]{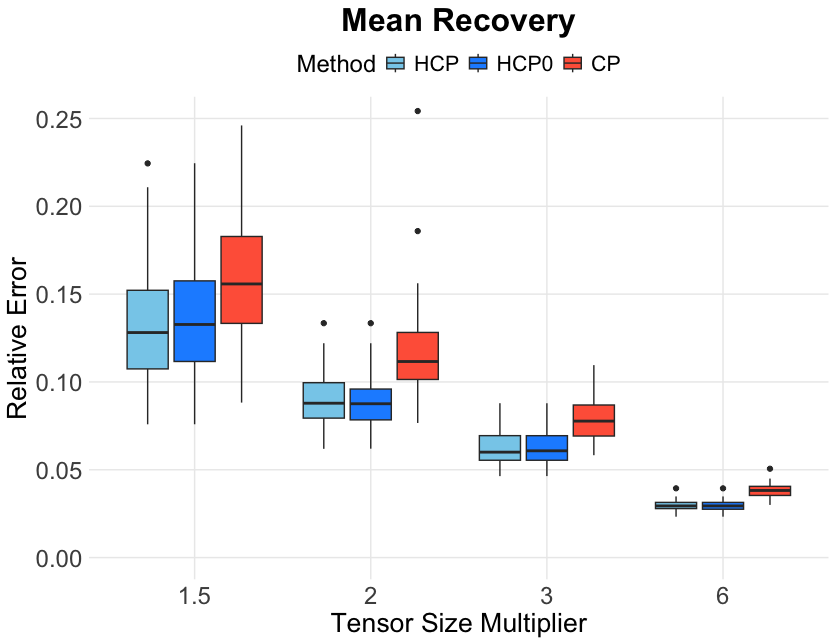}
        \caption{Mean recovery.}
        \label{fig:mean-sim-low}
    \end{subfigure}
    \hfill
    \begin{subfigure}[t]{0.49\linewidth}
        \centering
        \includegraphics[width=\linewidth]{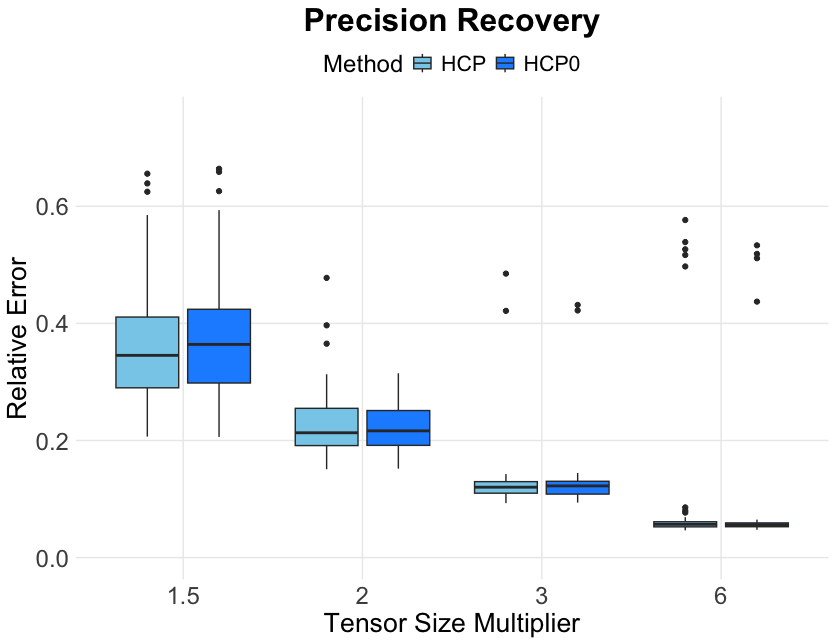}
        \caption{Precision recovery.}
        \label{fig:var-sim-low}
    \end{subfigure}
    \caption{Simulation results comparing recovery performance across tensor sizes. The left panel shows relative error in recovery of \(\mB\), while the right panel shows relative error in recovery of \(\mS\).}
    \label{fig:sim-results-low}
\end{figure}

Some takeaways from the rank misspecification results in Appendix~\ref{sec:appendix-simulation} are also worth noting, since in practical applications of HCP the true ranks will be unknown. First, it is generally more detrimental to underestimate the rank than it is to overestimate the rank, consistent with standard model selection discussions for CP \citep{core2003, koldaandbader09}. Secondly, we have key validation that for even poor choices of $R_{\mS}$, the mean estimate $\widehat \mB$ from HCP outperforms its standard CP counterpart.

\section{EEG Data Analysis} \label{sec:EEG}

In this section, we apply HCP to electroencephalogram (EEG) data, motivated by prior work using CP decompositions to recover interpretable factors from multi-channel EEG tensors \citep{cong2015tensor, Acar_2011_incomplete}.

For our EEG dataset, we  use the publicly available EEG Motor Movement/Imagery Dataset from PhysioNet \citep{Goldberger2000PhysioNet, PhysioNetEEGMMI}, collected using the BCI2000 system \citep{Schalk2004BCI2000}. This dataset contains 64-channel (i.e. electrode locations on the scalp) EEG recordings from subjects performing motor movement and imagery $2$-minute tasks. 
We consider $3$ ``runs'', with $15$ measurements each, of the first subject performing task $1$, where a target appears on the left or the right side of the screen, and the subject opens and closes their corresponding fist. Following  \cite{Acar_2011_incomplete} (with full details in Appendix~\ref{sec:appendix-EEG-preprocessing}) we computed Morlet wavelet time-frequency features for each channel and measurement over $4-40$ Hz in $2$ Hz increments, then took the magnitude of the complex coefficients and vectorized the frequency-by-time representation into a \(6099\)-dimensional feature vector. Hence, we obtain a tensor \(\mX \in \mathbb{R}^{64 \times 6099 \times 45}\), whose modes correspond to the \(64\) EEG channels, \(6099\) time-frequency coefficients, and \(45\) measurements.

We will analyze the low-rank recovery $\widehat{\mB}$ of $\mX$, computed using  \(R_{\mB}=3\), matching the rank used in \cite{Acar_2011_incomplete}. For HCP, we set \(R_{\mS}=1\), and assume $\mS$ is partially heteroskedastic, with constant variance along the measurement mode. Figure~\ref{fig:eeg-var} displays the estimated log-variance by time-frequency location, where we observe elevated average variance at lower frequencies. This matches a physical interpretation for EEG data and validates HCP as a viable alternative: low-frequency components are typically more variable due to the well-known \(1/f\)-like spectral structure of EEG \citep{Pritchard1992EEGFractal}, and they are also more susceptible to slow drift, motion artifacts, and ocular contamination \citep{Jung2000RemovingEEGArtifacts, Delorme2004EEGLAB}.

\begin{figure}[]
    \centering
    \includegraphics[width=0.6\textwidth]{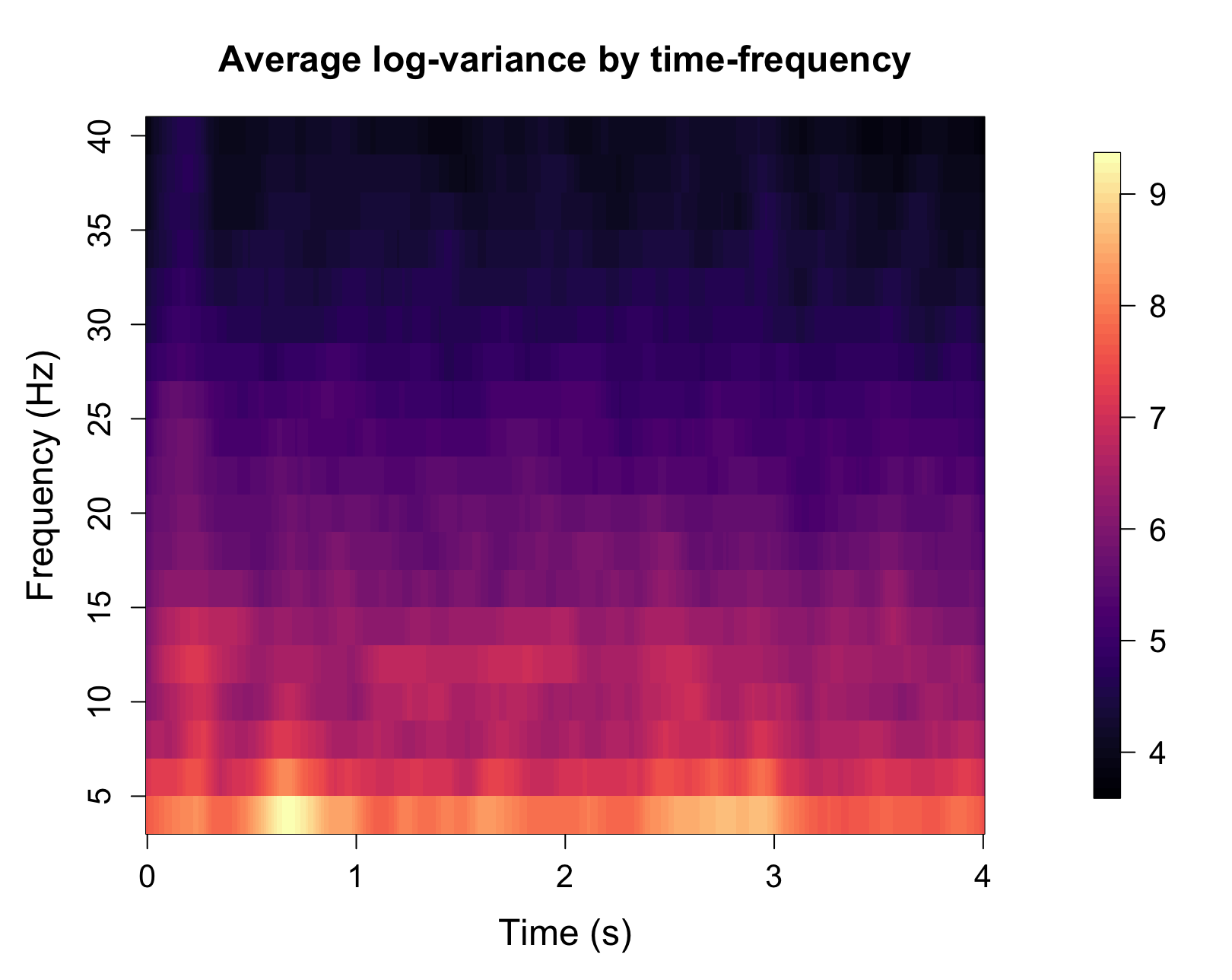}
    \caption{Estimated log-variance from HCP for the EEG tensor, averaged over channels and runs. The recovered variance is largest at lower frequencies.}
    \label{fig:eeg-var}
\end{figure}

The estimated mean factors for HCP and CP are presented in 
Figure~\ref{fig:eeg-brainfactors}. For both methods, the second component captures similar channel and time-frequency structure that is stable across measurements, and the first component appears to capture a low-frequency pattern that differs between measurement runs, though the measurement loadings are much more random in CP than in HCP. The third component displays the most notable difference between the methods: HCP identifies a localized channel pattern with loadings differing across runs, while CP appears to only capture an inverse of the low-frequency pattern in its first component. In Appendix~\ref{sec:appendix-EEG-preprocessing}, we examine a rank-$4$ fit for CP, where the additional component still only seems to model this low-frequency pattern. 

Coupled with our estimate of the variance structure in Figure~\ref{fig:eeg-var}, it seems likely that CP decomposition is treating the low-frequency variability as stable mean signal, while HCP, by modeling the heteroskedasticity separately, produces interpretable factors less dominated by this high-variance region. 

\begin{figure}[]
    \centering
    \includegraphics[width=\textwidth]{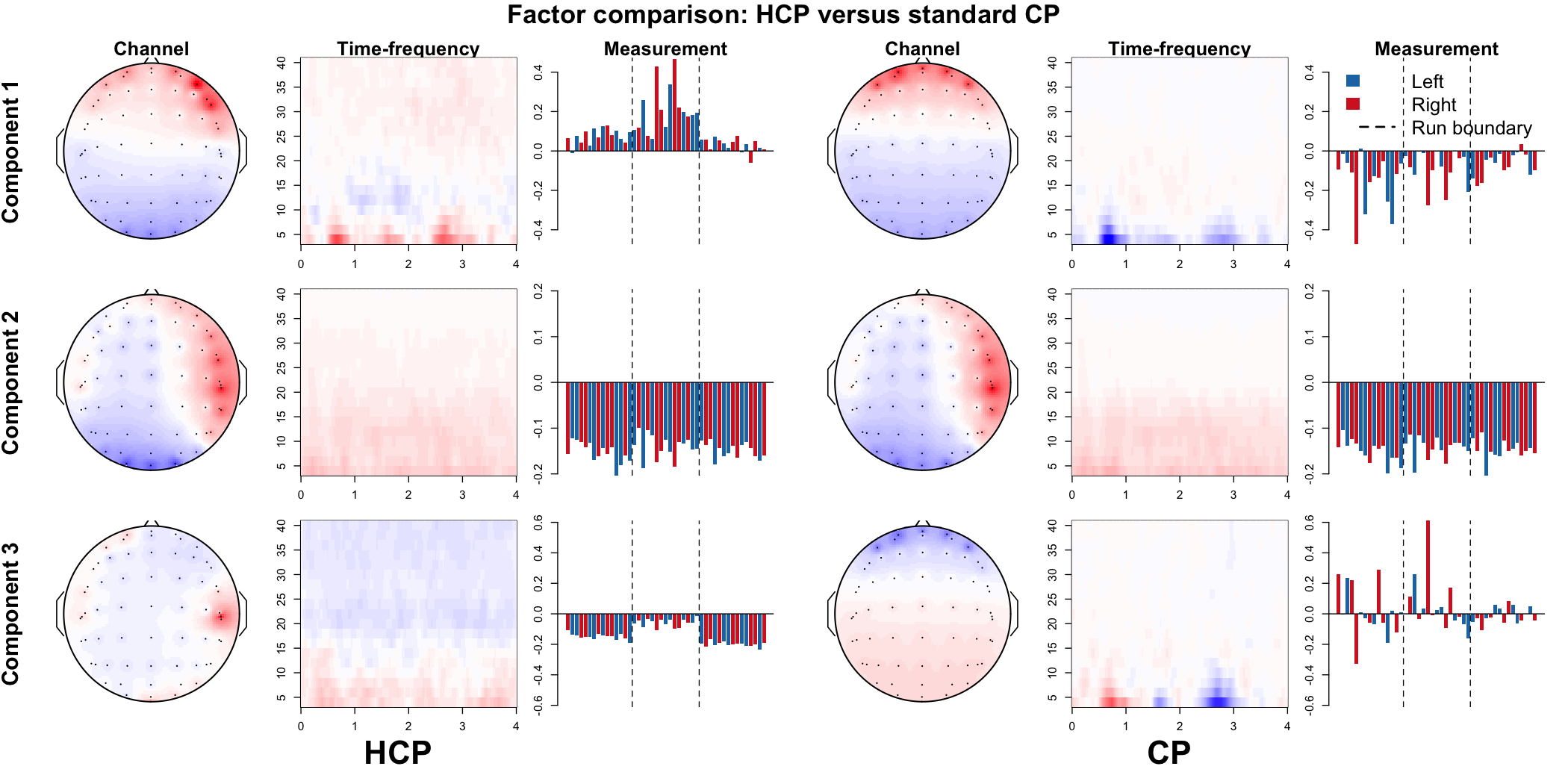}
    \caption{Visualization of the factor matrices recovered from the EEG tensor by CP and HCP with rows corresponding to components. Within each method, columns show the channel/scalp loadings, time-frequency loadings, and measurement loadings. Dashed vertical lines in measurement columns indicate recording-run boundaries.}
        \label{fig:eeg-brainfactors}
\end{figure}

\section{Discussion} \label{sec:conclusion}

In this paper, we introduced HCP, an extension of CP decomposition to an entrywise independent Gaussian model with unknown, low-rank precision tensor. We jointly estimate the mean and precision by alternating between the Weighted CP and Gamma-rate CP subproblems, deriving block-coordinate updates for each, and address the degeneracy of the objective with a residual modification that corresponds to a penalized likelihood approach.
Across synthetic experiments, we see that HCP is substantially more stable than CP under increasing heteroskedasticity, and we demonstrate the advantages of our ``two-parameter'' decomposition in handling model misspecification and providing interpretability benefits. 

Addressing the assumptions of our model provides many clear avenues for future work. For example, the entrywise independence could be replaced with a covariance from the tensor-multivariate-normal distribution \cite{MANCEUR201337tvn}, or the Gaussian generating distribution could be replaced with one allowing heavy tails, such as Student's t distribution, taking inspiration from \cite{auddy2021estimatingrankonespikedtensors}. Other two-parameter distributions may also be of interest. For example, the non-negative tensor factorizations developed by \cite{chiandkolda12} would likely not be adaptive to over-dispersed regimes, mirroring CP's degradation under heteroskedasticity. 
Finally, adaptations to other low-rank models for the parameters (e.g. Tucker, Tensor Train) could be considered.

\section*{Acknowledgements}

We thank Danny Dunlavy and Richard Lehoucq of Sandia National Labs, and Oscar Hernan Madrid Padilla of the University of California, Los Angeles for several helpful discussions during the writing of this manuscript. 

Sandia National Laboratories is a multi-mission laboratory managed and operated by National Technology \& Engineering Solutions of Sandia, LLC (NTESS), a wholly owned subsidiary of Honeywell International Inc., for the U.S. Department of Energy’s National Nuclear Security Administration (DOE/NNSA) under contract DE-NA0003525. This written work is authored by an employee of NTESS. The employee, not NTESS, owns the right, title, and interest in and to the written work and is responsible for its contents. Any subjective views or opinions that might be expressed in the written work do not necessarily represent the views of the U.S. Government. The publisher acknowledges that the U.S. Government retains a non-exclusive, paid-up, irrevocable, world-wide license to publish or reproduce the published form of this written work or allow others to do so, for U.S. Government purposes. The DOE will provide public access to results of federally sponsored research in accordance with the DOE Public Access Plan.

\newpage
\bibliographystyle{plainnat} 
\bibliography{references}

\newpage
\appendix

\section{Proof of Theorems~\ref{thm:GPCP} and~\ref{thm:penalty}} \label{sec:appendix-proofs}
In this section, we prove Theorems~\ref{thm:GPCP} and \ref{thm:penalty}. 

\paragraph{Theorem~\ref{thm:GPCP}}
Let $\bA, \bC \in \mathbb{R}_{\geq 0}^{m \times n}$, $\bB\in \mathbb{R}_{> 0}^{m \times p}$, and $\bOmega \in \mathbb{R}_{> 0}^{p \times n}$. Assume that each row of \(\bA\) has at least one positive entry and \(\bC\bOmega^\top\) is entrywise positive. Then the function 
$
f(\bB) = \bone_m^\top\left\{\bA*\log(\bB\bOmega) - \bC*(\bB\bOmega)
\right\}\bone_n
$
is non-decreasing under the multiplicative update
\begin{equation} 
    \bB\leftarrow \bB* \left\{ \left [ \bA\oslash (\bB\bOmega)\right]\bOmega^\top\right\} \oslash (\bC\bOmega^\top).
\end{equation}

\begin{proof}
First, write the function $f$ entrywise as
$$
f(\bB) = 
\sum_{ij} \left[ a_{ij}\log\left(\sum_r b_{ir}\omega_{rj}\right)-c_{ij}\sum_r b_{ir}\omega_{rj} \right], 
$$
and the proposed update rule as 
\begin{equation} \label{eq:entry-update}
    b_{ir} \gets \frac{\bar b_{ir}}{\sum_{j}\omega_{rj}c_{ij}} \left( \sum_j \omega_{rj}\frac{a_{ij}}{\sum_r \bar b_{ir}\omega_{rj}} \right).
\end{equation}

Now, let $\alpha_{rij} = \dfrac{\bar b_{ir}\omega_{rj}}{\sum_r \bar b_{ir}\omega_{rj}} $ and consider the function 
$$
g(\bB,\bar \bB) = 
\sum_{ijr}\left[
\alpha_{rij}a_{ij}\log\left(\dfrac{b_{ir}\omega_{rj}}{\alpha_{rij}}\right)- c_{ij}b_{ir}\omega_{rj}
\right].
$$

Note, first, that 
\begin{align*}
        g(\bB, \bB) &= \sum_{ijr}\left[
\alpha_{rij}a_{ij}\log\left(\dfrac{b_{ir}\omega_{rj}}{\alpha_{rij}}\right)- c_{ij}b_{ir}\omega_{rj}
\right] \\
&= \sum_{ij} \frac{a_{ij}}{\sum_{r} b_{ir}\omega_{rj}} \left[
\sum_{r} b_{ir}\omega_{rj}\log\left(\sum_{r}  b_{ir}\omega_{rj}\right)  \right]- \sum_{ij}c_{ij}\sum_{r} b_{ir}\omega_{rj} \\
&= \sum_{ij} a_{ij} \log\left(\sum_{r} b_{ir}\omega_{rj}\right) - \sum_{ij}c_{ij}\sum_{r} b_{ir}\omega_{rj} = f(\bB). 
\end{align*}

Under the assumptions \(\bar{\bB}>0\) and \(\bOmega>0\), we have \(\alpha_{rij}>0\) for all \(i,j,r\), and \(\sum_r \alpha_{rij}=1\). 
The concavity of $\log$ gives us  $\log\left(\sum_r b_{ir}\omega_{rj}\right) \geq \sum_r \alpha_{rij}\log\left(\dfrac{b_{ir}\omega_{rj}}{\alpha_{rij}}\right)$. So, $g(\bB,\bar \bB) \leq f(\bB)$. This, along with $g(\bB, \bB) = f(\bB)$, means $g(\bB,\bar \bB)$ minorizes $f(\bB)$. 

So, we consider the basic minorization-maximization iteration $\bB_{k+1} = \argmax_{\bB} g(\bB, \bB_k)$, which will be non-decreasing with respect to $f$ since $f(\bB_{k+1}) \geq g(\bB_{k+1}, \bB_k) \geq g(\bB_k, \bB_k) = f(\bB_k)$. The first derivative of $g(\bB,\bar \bB)$ is
$$
\dfrac{\partial g(\bB,\bar \bB)}{\partial b_{ir}} = 
b_{ir}^{-1}\sum_j(\alpha_{rij}a_{ij}) - \sum_j(\omega_{rj}c_{ij}) .
$$
Setting the above to zero leads to
$
b_{ir} = \sum_j(\alpha_{rij}a_{ij})/\sum_j(\omega_{rj}c_{ij}),
$
matching the entrywise version of our update rule (equation~\ref{eq:entry-update}). 

\end{proof}

\paragraph{Theorem~\ref{thm:penalty}.}
Consider the scalar $\delta > 0$. The log-likelihoods $\ell_{\delta}(\mS| \mX, \mB)$ (equation~\ref{eq:penalized-likelihood}) and $\ell(\mS| \mX, \mB)$ (equation~\ref{eq:likelihood_precision}) are related as 
$$\ell_{\delta}(\mS| \mX, \mB) =  l(\mS| \mX, \mB) - \frac{\delta}{2} \| \mS\|_{1}.$$

\begin{proof}
\begin{align*}
    \ell_{\delta}(\mS \mid\mX, \mB) &\coloneqq \dfrac{1}{2} \mathbf{1}_{N_d}^\top \left(\log(\widetilde\bS_d\bOmega_{d}) - ( \mR_{(d)} + \delta) * \widetilde\bS_d\bOmega_{d} \right) \mathbf{1}_{N_{-d}} \\
    &= \dfrac{1}{2} \mathbf{1}_{N_d}^\top  \left(  \log(\widetilde\bS_d\bOmega_{d}) - \mR_{(d)} * \widetilde\bS_d\bOmega_{d} -  \delta \widetilde\bS_d\bOmega_{d}  \right) \mathbf{1}_{N_{-d}} \\
    &= l(\mS \mid\mX, \mB) - \frac{\delta}{2} \sum_{i, j} (\widetilde\bS_d \bOmega_{d})_{i, j} = l(\mS| \mX, \mB) - \frac{\delta}{2} \| \mS\|_{1}.
\end{align*} 
The final equality holds since all elements of $\mS$ are positive by definition.
\end{proof}

\section{A Short Primer on Tensor Operations} \label{sec:appendix-background}

In this section, we  provide a brief review of tensor and matrix product definitions and properties we use in the main manuscript. See \cite{kolda06} for a more comprehensive review. We will use the $D$-way tensor $\mX \in \mathbb{R}^{N_1 \times \ldots \times N_D}$.

Fixing all indices of $\mX$ except one yields a \emph{fiber}, the higher-order analogue of a matrix row/column \cite{kolda06}. Fixing all indices except two yields a \emph{slice}, i.e., a two-dimensional section of the tensor. 
The \emph{mode-$d$ matricization} $\mX_{(d)}$ (see \cite{kolda06} for the general case), 
 arranges the mode-$d$ fibers as the columns of a matrix, so $\mX_{(d)} \in \mathbb{R}^{N_d \times N_{-d}}.$ 

The tensor $\mX$ is \emph{rank-one} if it can be written as the outer product of $D$ vectors,
$
\mX = \bb^{(1)} \circ \bb^{(2)} \circ \cdots \circ \bb^{(D)},
$
meaning $x_{i_1\cdots i_d} = \bb^{(1)}_{i_1}\cdots \bb^{(D)}_{i_d}$. Tensor rank, denoted  as $\text{rank}(\mX)$, is defined as the smallest integer $R$ such that $\mX$ can be written as a sum of $R$ rank-one tensors \cite{kolda06}. When referring to non-negative tensors such as $\mS$, we will often use the non-negative rank \citep{limandcomon09}, denoted $\text{rank}_{+}(\mS)$, and defined as the smallest integer $R$ such that $\mS$ can be written as a sum of $R$ rank-one tensors, where each rank-one tensor has non-negative factors. Note that it does not generally hold that $\text{rank}(\mS) = \text{rank}_{+}(\mS)$.

When a CP-rank decomposition is known, we can express the tensor in terms of its factor matrices using the \textit{Kruskal operator}. Specifically, 
\begin{equation}
    \mathcal{B} = \sum_{r=1}^{R} \lambda_r\mathbf{b}_r^{(1)} \circ \mathbf{b}_r^{(2)} \circ \cdots \circ \mathbf{b}_r^{(D)} \coloneqq [\![\blambda;  \bB_1, \bB_2, \ldots, \bB_D ]\!],
\end{equation}
where $\mathbf{b}_r^{(d)} \in \mathbb{R}^{N_d}$ are the factor column vectors for mode $d$, making up each rank-one component, and $\bB_d = [\mathbf{b}_1^{(d)}, \mathbf{b}_2^{(d)}, \ldots, \mathbf{b}_R^{(d)}] \in \mathbb{R}^{N_d \times R}$ is the factor matrix for mode $d$. 

As mentioned in Section \ref{sec:background}, Kruskal tensors have an inherent scaling ambiguity, hence our decision to write the scaling vector $\lambda$ explicitly, and use factor matrices with unit column norms. We use the following function in Algorithms~\ref{alg:heterocp-mean}, \ref{alg:heterocp-precision}, and \ref{alg:heterocp} to normalize a factor matrix and return the scaling vector. 

\begin{algorithm}[H]
\caption{\textsc{ColumnNormalize}}
\label{alg:column-normalize}
\begin{algorithmic}[1]
\State \textbf{Input:}  Matrix \(\bB \in \mathbb{R}^{N\times R}\), norm \(p\geq 1\)
\For{$r = 1,\ldots,R$}
    \State \(\lambda_{r} \gets \|\bB_{:,r}\|_p\)
    \State \(\bB_{:,r} \gets \bB_{:,r}/\lambda_{r}\)
\EndFor
\State \Return \((\bB,\blambda)\)
\end{algorithmic}
\end{algorithm}

The Kronecker product of matrices $\bA \in \mathbb{R}^{I\times J}$ and $\bB \in \mathbb{R}^{K\times L}$ is denoted by $\bA \otimes \bB $. The result is a matrix of size $(IK)\times (JL)$ and defined by
\[
\bA \otimes \bB =
\begin{bmatrix}
\bA_{11}\bB & \bA_{12}\bB & \cdots & \bA_{1J}\bB\\
\bA_{21}\bB & \bA_{22}\bB & \cdots & \bA_{2J}\bB\\
\vdots  & \vdots  & \ddots & \vdots\\
\bA_{I1}\bB & \bA_{I2}\bB & \cdots & \bA_{IJ}\bB
\end{bmatrix}.
\]

The Khatri--Rao product is the ``matching columnwise'' Kronecker product. Given matrices
$\bA \in \mathbb{R}^{I\times K}$ and $\bB \in \mathbb{R}^{J\times K}$, their Khatri--Rao product is denoted by $\bA \odot \bB$. The result is a matrix of size $(IJ)\times K$ defined by
\[
\bA \odot \bB =
\begin{bmatrix}
\ba_1 \otimes \bb_1 & \ba_2 \otimes \bb_2 & \cdots & \ba_K \otimes \bb_K
\end{bmatrix},
\]
where $\ba_k$ and $\bb_k$ denote the $k$th columns of $\bA$ and $\bB$, respectively. If $\ba$ and $\bb$ are vectors, then the Khatri--Rao and Kronecker products are identical, i.e., $\ba \otimes \bb = \ba \odot \bb$.

Finally, we leveraged the following propositions in manipulating the log-likelihood equations in Sections~\ref{sec:mean} and~\ref{sec:precision}. 
\begin{proposition}{(Proposition 5.3(b) in \cite{kolda06})} \label{prop:kruskal-decomp}  The mode-$d$ matricization of  Kruskal tensors:
\begin{equation*}
[\![ \blambda;\bX_1, \bX_2, \ldots, \bX_D ]\!]_{(d)} = \bX_d\diag(\blambda) \left( \bX_D \odot \cdots \odot \bX_{d+1} \odot \bX_{d-1} \odot \cdots \odot \bX_1 \right)^{\top}.
\end{equation*}
\end{proposition}
\begin{proposition}{(Equation (520) in \cite{petersen2008})} \label{prop:vec-trick} $\vecc(\bA\bX\bB^\top) = (\bB\otimes \bA)\vecc(\bX)$. 
\end{proposition}

\section{Heteroskedastic PCA} \label{sec:appendix-HPCA}

As mentioned in the main manuscript, heteroskedasticity has been much more extensively studied for matrix principal component analysis (PCA) \citep{hong2022optimallyweightedpcahighdimensional, Hong_PPCAT_2021, zhang2021heteroskedasticpcaalgorithmoptimality} or for (low-rank) matrix estimation / matrix denoising \citep{Agterberg_2022, cai2020subspaceestimationunbalancedincomplete, Leeb_2021, zhang2024matrixdenoisingdoublyheteroscedastic, su2024datadrivenoptimalshrinkagesingular, behne2022fundamentallimitsrankonematrix}. 

Modern heteroskedastic PCA formulations leverage what we call ``partial heteroskedasticity'' (defined in Section~\ref{sec:partial-hetero}), where only some directions have heteroskedastic data. \cite{zhang2021heteroskedasticpcaalgorithmoptimality, zhouDeflated10.1214/24-AOS2456} model row-wise heteroskedasticity, which they call ``within-sample'' heteroskedasticity, since a sample is a column in the observation matrix. In a contrasting regime, 
\cite{Hong_PPCAT_2021, Hong_PCA_2016, HONG_PCA_2018435, hong2022optimallyweightedpcahighdimensional}, model heteroskedasticity in PCA ``between-samples'' or column-wise, where there are distinct noise groups and group memberships known a priori. Among these works, \cite{Hong_PPCAT_2021} is closest to our setting in that estimating noise variances, rather than only obtaining robustness to heteroskedasticity, is a primary goal. Although its likelihood has an analogous degeneracy, the Gaussian prior on the low-dimensional coefficients regularizes the low-rank fit, so the variance update is not based on residuals from a freely optimized deterministic reconstruction as in HCP. The degeneracy only occurs for their optimization problem when a sample lies exactly in the learned low-dimensional subspace, which, as in the rank-one case for HCP discussed in Section~\ref{sec:rank-one}, is unlikely for noisy data.

For matrix estimation/denoising, row-wise heteroskedasticity is popular \citep{Agterberg_2022, cai2020subspaceestimationunbalancedincomplete, su2024datadrivenoptimalshrinkagesingular, Leeb_2021}, though \cite{behne2022fundamentallimitsrankonematrix} considers block-dependent signal strengths across pre-specified groups, for blocks that can be both row and column dependent, and \cite{zhang2024matrixdenoisingdoublyheteroscedastic} considers a more general structured-noise model of the form \(\bW=\Xi^{1/2}\bZ\Sigma^{1/2}\), which allows heterogeneity and correlations along both matrix dimensions. 

This final variance structure can be viewed 
as the order-\(2\) (matrix) specialization of the Kronecker-separable 
covariance structure used for Tensor-on-Tensor Regression 
\cite{llosaandmaitra22}, and corresponds to a special case of HCP 
(Model~\ref{eq:generative-model}) when the precision tensor \(\mS\) is 
rank-one. In particular, if the matrix noise model is written as 
\(\bW=\Xi^{1/2}\bZ\Sigma^{1/2}\), with 
\(\bZ_{ij}\stackrel{\text{iid}}{\sim}\mathcal{N}(0,1)\), 
\(\Xi=\operatorname{diag}(\xi_1,\ldots,\xi_d)\), and 
\(\Sigma=\operatorname{diag}(\sigma_1,\ldots,\sigma_d)\), then 
\(\bW_{ij}\sim \mathcal{N}(0,\bS_{ij}^{-1})\), for $\bS_{ij} = \xi_i^{-1}\sigma_j^{-1}$. This comparison primarily highlights the divergence, though, of Model~\ref{eq:generative-model} from the heteroskedastic matrix estimation literature: in HCP, the noise is generated entrywise independently, and the structured variability is encoded through the low-rank form of the precision tensor \(\mS\), rather than through an explicit covariance tensor. As mentioned in Section~\ref{sec:conclusion}, a heteroskedastic CP decomposition that parametrizes the variance structure with this Kronecker-separable covariance to allow for explicit dependence presents an interesting avenue for future work.

\section{Full Procedure for Heteroskedastic CP Decomposition} \label{sec:appendix-heterocp}

In this section we present the full procedure for Heteroskedastic CP Decomposition, without relying on calls to the Weighted CP or Gamma-rate CP Algorithms. We include the details for the $R_{\mS} = 1$ case,  partial heteroskedasticity, and the constrained residuals, discussed in Sections~\ref{sec:rank-one}, \ref{sec:partial-hetero}, and \ref{sec:constraint}, respectively. 

As in Section~\ref{sec:partial-hetero}, we assume homoskedastic modes will always be the final $Q$ modes. To accommodate other choices, one should simply permute the tensor $\mX$ before running Algorithm~\ref{alg:heterocp-full}, and reorder the returned factor matrices after Algorithm~\ref{alg:heterocp-full} concludes. 

\begin{algorithm}[H]
\caption{Heteroskedastic CP Decomposition}
\label{alg:heterocp-full}
\begin{algorithmic}[1]
\State \textbf{Input:} Data \( \mX \in \mathbb{R}^{N_1 \times \cdots \times N_D} \)
\State \textbf{Input:} Ranks \(R_{\mB}\), \(R_{\mS}\)
\State \textbf{Input:} Number of Heteroskedastic modes \(P\leq D\)
\State \textbf{Input:} Residual constraint \(\delta\geq 0\) \Comment{Default value is \(\delta = \sqrt{\epsilon_{\mathrm{mach}}} \)}
\State \textbf{Input:} Maximum iterations \(i_{\mathrm{tot}}, i_{\mB}, i_{\mS}, i_{\mathrm{in}}\)
\State \textbf{Initialize:} 
\(\{\bB_1,\ldots,\bB_D\}\), with unit column \(\ell_2\) norms 
\State \textbf{Initialize:} 
\(\{\blambda_{\mS},\bS_1,\ldots,\bS_P\}\), entrywise positive, with unit column \(\ell_1\) norms for $\bS_d$
\State Define $\bS_q := \bone_{N_q}\bone_{R_{\mS}}^\top$ for $q = P + 1, \ldots D$ \Comment{``Factor Matrices'' for the Homoskedastic Modes} 

\While{not converged and iterations \(\leq i_{\mathrm{tot}}\)}
    
    \State $\mZ = \mX*[\![ \blambda_{\mS}; \bS_1,\ldots,\bS_D]\!]$  \Comment{Weighted CP}  
    \For{\(t_{\mB}=1,\ldots,i_{\mB}\)} 
        \State $\bH_{r, d} \gets \left(\bB_d^\top
                \operatorname{diag}\!\left((\bS_d)_{:,r}\right)
                \bB_d\right)$ for \(r=1,\ldots,R_{\mS}\), \(d=1,\ldots,D\)

        \For{\(d = 1,\ldots,D\)}
       
            \State $\bZ_d
                \gets
                {\mZ}_{(d)} \left(\odot_{c\neq d}\bB_c\right)     $            \Comment{MTTKRP}  
            
            \State $\bH_{r, -d} \gets \left( \underset{c \neq d}{*} \bH_{r, c} \right)$  
            \For{\(i_d = 1,\ldots,N_d\)}

                \State Define $\bG
                    \gets
                    \sum_{r=1}^{R_{\mS}}
                    (\blambda_{\mS})_r
                    (\bS_d)_{i_d,r} \bH_{r, -d} $ \Comment{Factored Weights Optimization}
                \State \((\bB_d)_{i_d,:} \gets \bb^\top\), where \(\bb\) solves $\bG\bb = (\bZ_d)_{i_d,:}^\top $
                
            \EndFor
    
            \State \((\bB_d,\blambda_{\mB})\gets
            \Call{\textsc{ColumnNormalize}}{\bB_d,2}\)

            \State $\bH_{r, d} \gets \left(\bB_d^\top
                \operatorname{diag}\!\left((\bS_d)_{:,r}\right)
                \bB_d\right)$ for \(r=1,\ldots,R_{\mS}\)
    
        \EndFor
    \EndFor
    \State $\mR
        \gets
        \left(
            \mX-[\![\blambda_{\mB};\bB_1,\ldots,\bB_D]\!]
        \right)^{*2}$
    \If{\(P<D\)}
        \State $(\mR)_{\bi}
            \gets
            \frac{1}{\prod_{q=P+1}^{D}N_q}
            \sum_{\bj\in [N_{P+1}]\times\cdots\times[N_D]}
            (\mR)_{(\bi,\bj)}$ for \(\bi\in [N_1]\times\cdots\times[N_P]\)
    \EndIf   
    \State $\mR \gets \mR + \delta$
    
    \If{\(P=1\)} \Comment{Special Case: One Heteroskedastic Mode}
        \State \((\bS_1,\blambda_{\mS})\gets
        \Call{\textsc{ColumnNormalize}}{\mR^{*-1},1}\)
    
    \Else \Comment{Gamma-rate CP}
        \For{\(t_{\mS}=1,\ldots,i_{\mS}\)}
            \For{\(d=1,\ldots,P\)}
                \If{\(R_{\mS}=1\)} 
                    \State $\bOmega_d
                    \gets
                    \left(
                        \otimes_{c\leq P,\;c\neq d}\bS_c
                    \right)
                    $ \Comment{Note $\bS_c$ are vectors when $R_{\mS} = 1$}
                    \State $\widetilde \bS_d \gets \Big( \prod_{c \leq P, c \neq d} N_c\Big)
                    \left(
                        (\mR)_{(d)}\bOmega_d
                    \right)^{*-1}$ \Comment{Closed Form Update for $R_{\mS} = 1$}
                \Else
                    \State $\bOmega_d
                    \gets
                    \left(
                        \odot_{c\leq P,\;c\neq d}\bS_c
                    \right)^\top
                    $ 
                    \State $\bD_d
                        \gets
                        (\mR)_{(d)}\bOmega_d^\top $ \Comment{MTTKRP}  
                    \State $\widetilde{\bS}_d
                        \gets
                        \bS_d\operatorname{diag}(\blambda_{\mS})$
                    \For{\(i=1,\ldots,i_{\mathrm{in}}\)}
                        \State $\widetilde{\bS}_d
                            \gets
                            \widetilde{\bS}_d
                            *
                            \left(
                                (\widetilde{\bS}_d\bOmega_d)^{*-1}
                                \bOmega_d^\top
                            \right)
                            \oslash
                            \bD_d $
                    \EndFor
                \EndIf
                \State $(\bS_d,\blambda_{\mS})
                        \gets
                        \Call{\textsc{ColumnNormalize}}{\widetilde{\bS}_d,1}$
            \EndFor
        \EndFor
    \EndIf

\EndWhile

\State \textbf{Output:} 
\(\{\blambda_{\mB},\bB_1,\ldots,\bB_D\}\) and
\(\{\blambda_{\mS},\bS_1,\ldots,\bS_P\}\)

\end{algorithmic}
\end{algorithm}

\section{Computational Complexity} \label{sec:appendix-complexity}

In this section, we provide the derivation of the time complexity of Algorithm~\ref{alg:heterocp} given in Equation~\ref{eq:complexity}, and provide the space complexity as well.  Lower-order costs associated with column normalization, storing scaling vectors, and convergence checks are omitted. Recall $N_{\times} \coloneqq  \prod_{d=1}^D N_d, \, N_+ \coloneqq  \sum_{d=1}^D N_d,$ and $N_{-d} \coloneqq  \prod_{c\neq d} N_c = N_{\times}/N_d.
$

First, we derive the cost for Algorithm~\ref{alg:heterocp-mean}. For mode \(d\), the update is performed row-wise, and for each row \(k=1,\ldots,N_d\), the update solving the weighted least-squares problem  (Equation \eqref{eq:mean-update}) by a normal-equation, QR, SVD, or pseudoinverse-based calculation requires, to leading order, $O(N_{-d}R_{\mB}^2 + R_{\mB}^3)$ 
operations for each row. Since there are \(N_d\) rows in mode \(d\), the cost for updating \(\bB_d\) is
\[
O\left(N_dN_{-d}R_{\mB}^2 + N_dR_{\mB}^3\right)
=
O\left(N_{\times}R_{\mB}^2 + N_dR_{\mB}^3\right).
\]
Summing over \(d=1,\ldots,D\), one full sweep of Algorithm~\ref{alg:heterocp-mean} costs $O\left(
DN_{\times}R_{\mB}^2
+
N_+R_{\mB}^3
\right).$
If $R_{\mB} \ll N_{-d} ~ \forall ~ d \in [D]$, typical for CP decomposition, this reduces to $O(DN_{\times}R_{\mB}^2)$, the cost reported in Section \ref{sec:mean}. 

For reference, a standard dense CP-ALS sweep with rank \(R\) has cost $O(DN_{\times}R + N_+R^2 + DR^3)$, though we only reported \(O(DN_{\times}R)\) in Section \ref{sec:mean}, since the MTTKRP (Matricized Tensor times Khatri--Rao Product) cost will dominate the Gram and solve terms whenever
\[
    R \ll \frac{DN_{\times}}{N_+}
    \quad\text{and}\quad
    R^2 \ll N_{\times}, 
\]
which is almost always the case for CP decomposition. The computational advantage of CP-ALS comes from the Khatri--Rao Gram 
identity
\[
    \bPi_{d}\bPi_{d}^\top
    =
    \underset{c\neq d}{*}
    (\bB_c^\top\bB_c),
\]
which reduces the Gram calculation from the large 
\(N_{-d}\times R\) Khatri--Rao matrix \((\bPi_{d})^\top\) to an 
\(R\times R\) matrix. In contrast, the Gram matrix blocks $(\bG_{d, i_d}) = \bPi_{d}\diag((\mS_{(d)})_{i_d, :})\bPi_{d}^\top$
used for Algorithm \ref{alg:heterocp-mean} depend on $i_d \in [N_d]$, and do not form one shared  \(R_{\mB}\times R_{\mB}\) system for the entire mode update.

However, as pointed out in Section \ref{sec:mean}, when the weights \(\mS\) are low-rank, as assumed in this work, the HCP Gram matrices can be assembled more efficiently using Equation~\eqref{eq:fast-gram}. Now, the row wise Gram blocks can be assembled as  
\[
    \bG_{d,i_d}
    =
    \sum_{r=1}^{R_{\mS}}
    (\blambda_{\mS})_r
    (\bS_d)_{i_d,r}
    \left(\underset{c\neq d}{*} \bH_{r,c} \right),
\]
where 
$$
\bH_{r, c} = 
        \bB_c^\top    \diag\big((\bS_c)_{:,r}\big)    \bB_c
     \in \mathbb{R}^{R_{\mB}\times R_{\mB}}, \qquad r \in [R_{\mS}], \qquad c \in [D]
$$
costs only \(O(N_+R_{\mB}^2R_{\mS})\), rather than 
\(O(DN_{\times}R_{\mB}^2)\). Computing the Hadamard contractions costs \(O(D^2R_{\mS}R_{\mB}^2)\), though this can be reduced to \(O(DR_{\mS}R_{\mB}^2)\) by using cached prefix and suffix products. In either case, the term is lower order.


Thus, whenever
\begin{equation} \label{eq:small-rank}
    R_{\mB}R_{\mS} \ll \frac{DN_{\times}}{N_+}
    \quad\text{and}\quad
    R_{\mB}^2 \ll \frac{DN_{\times}}{N_+},  
\end{equation}
the \(O(DN_{\times}R_{\mB})\) term in the normal equations (Equation~\eqref{eq:mean-update}) is dominant. 
Under these rank conditions, Algorithm~\ref{alg:heterocp-mean} has the same leading dense-data order as a rank-\(R_{\mB}\) CP-ALS sweep. 

Next, let us consider Algorithm~\ref{alg:heterocp-precision}. The dense matrix products involving \(\bOmega_{d}\) each cost $O(N_dN_{-d}R_{\mS})
=
O(N_{\times}R_{\mS}).$ Thus, a full sweep over all modes, with \(i_{\text{in}}\) inner multiplicative updates, costs $O(Di_{\text{in}}N_{\times}R_{\mS}).$ 

For partially heteroskedastic precision, define $N_{\times, P} \coloneqq \prod_{d=1}^P N_d$ and $N_{+, P} \coloneqq  \sum_{d=1}^{P}N_d$ where $P < D$ is the number of heteroskedastic modes (like in Section~\ref{sec:partial-hetero}, we assume the tensor has been permuted such that heteroskedastic modes are the first $P$ modes). The leading cost of a precision sweep then becomes $O\left(
i_{\text{in}}P N_{\times, P} R_{\mS}
\right)$, which can be a substantial reduction from $O\left(i_{\text{in}}
D N_{\times} R_{\mS}
\right)$ when \(N_{\times, P}\ll N_{\times}\).

In addition, many practical applications (such as our EEG data analysis in Section~\ref{sec:EEG}) will only require $R_{\mS} = 1$. As we observed in Section~\ref{sec:rank-one}, in this rank-one case the factor update in Algorithm~\ref{alg:heterocp-precision} has a closed form, so we don't need any inner multiplicative iterations. This reduces the complexity from $O(Di_{\text{in}}N_{\times}R_{\mS})$ to $O(DN_{\times}R_{\mS})$. 

Combining the mean and precision updates, a single sweep for each subroutine of Algorithm~\ref{alg:heterocp} has computational complexity
\[
O\left(
D N_{\times} R_{\mB}
+
i_{\text{in}} D N_{\times} R_{\mS}
\right),
\]
where we are assuming the ranks satisfy Equation \eqref{eq:small-rank}. For the full alternating procedure, run for \(i_{\text{tot}}\) outer iterations, each with \(i_{\mB}\) sweeps in Algorithm~\ref{alg:heterocp-mean} and \(i_{\mS}\) sweeps in Algorithm~\ref{alg:heterocp-precision}, the total time complexity is
\[
O\left(
i_{\text{tot}}
\left[
i_{\mB}
D N_{\times} R_{\mB}
+
i_{\mS} i_{\text{in}} D N_{\times} R_{\mS}
\right]
\right).
\]

Since \(i_{\text{in}}\) is typically small, \(R_{\mS}\) is often smaller than \(R_{\mB}\), and partial heteroskedasticity can further reduce the
precision-update cost, we say that a single sweep for Algorithm~\ref{alg:heterocp} has roughly the same cost as a sweep for CP-ALS in the small-rank
regimes satisfying Equation~\eqref{eq:small-rank} targeted here. This makes our method highly competitive from a time complexity standpoint. 

Finally, consider space complexity. Storing the dense data/parameter tensors and the factor matrices requires
\[
O\left(
N_{\times}
+
N_{+}R_{\mB}
+
N_{+}R_{\mS}
\right).  
\]
In addition, the $\bH_{r, c}$ terms (or the prefix and suffix cached terms) for the Hadamard contractions require an additional workspace $O(DR_{\mS}R_{\mB}^{2})$. 

If the Khatri--Rao matrices \(\bPi_d\) and \(\bOmega_d\) appearing in Algorithms~\ref{alg:heterocp-mean} and~\ref{alg:heterocp-precision} are explicitly materialized, then the additional workspace is
\[
O\left(
\max_{d} N_{-d}(R_{\mB}+R_{\mS})
\right).
\]
If the Khatri--Rao contractions are implemented in a streamed or batched manner, the term \(\max_d N_{-d}(R_{\mB}+R_{\mS})\) can be reduced to the corresponding batch workspace. 
Our implementation does not leverage this streaming optimization, so the space bound for the implementation used in this paper is
\[
O\left(
N_{\times}
+
N_{+}R_{\mB}
+
N_{+}R_{\mS}
+
DR_{\mS}R_{\mB}^{2}
+
\max_{d}N_{-d}(R_{\mB}+R_{\mS})
\right).
\]

\section{Model Selection} \label{sec:model-selection}

In this section we give our recommendations for how to choose the user specified parameters in Algorithms~\ref{alg:heterocp-mean},~\ref{alg:heterocp-precision}, and~\ref{alg:heterocp}, namely the ranks, convergence criteria, maximum number of iterations, and initialization mechanism. We also discuss choosing homoskedastic modes and the size of the constraint $\delta$. 

In the examples in the main manuscript, we have treated the ranks $R_{\mB}$ and $R_{\mS}$ as true and known, or tailored the ranks to the application. In Appendix~\ref{sec:appendix-simulation} we demonstrate that HCP's performance will typically be strong as long as the ranks are large enough, and that misspecification of the precision rank will not drag HCP's performance below that of standard CP decomposition. However, a data-driven mechanism for model selection is still desirable. 

The simplest evaluation metrics would be to simply use RSS $\|\mX - \widehat{\mB} \|_{F}^2$ or Relative Error $\frac{\|\mX - \widehat{\mB} \|_{F}}{\|\mX\|_{F}}$, or their $\widehat{\mS}$ weighted counterparts. However, these typically decrease with ranks $R_{\mB}$ and $R_{\mS}$, so they are not great measures of ``goodness-of-fit'' for model selection. In their textbook on tensor decompositions, \cite{Ballard_Kolda_2025} advocate for an elbow-style approach of selecting smallest $R$ for which $R+1$ only offers minor improvement in RSS or relative error, which we can adapt to HCP by comparing the present (weighted, if desired) relative error of $R_\text{mean}, R_\text{var}$ to both $R_\text{mean}+1, R_\text{var}$ and $R_\text{mean}, R_\text{var}+1$. CORCONDIA \cite{core2003} or NORMO \cite{FERNANDES2020103926} propose ``appropriateness'' measures for goodness of fit, which rely on structural interpretations of the factor matrices. 

Since we take a likelihood-based approach, we can leverage Akaike information criterion (AIC) or Bayesian information criterion (BIC) to select the model. A BIC-style heuristic, for example, would simply be 
\begin{equation} \label{eq:BIC}
    -2\ell(\mB, \mS \mid\mX) + \log(N_{\times})\left[(R_{\mB}+R_{\mS})N_+ - (R_{\mB}+R_{\mS})(D-1) \right],
\end{equation}
where parameter counts take away $D-1$ degrees of freedom per component due to the scaling non-identifiability discussed in Section~\ref{sec:background} (note that we are assuming fully heteroskedastic precision, here). Many other likelihood-based tensor methods leverage AIC or BIC similarly, see \cite{Wang2023} for AIC or \cite{zhang2014tensorgeneralizedestimatingequations} and \cite{KetensorQR2003} for BIC. Technically, Equation~\eqref{eq:BIC} is only BIC when we are using the exact maximum likelihood in Equation~\ref{eq:joint-likelihood}, so, since in practice we optimize the penalized likelihood $\ell_\delta$, we only call this a BIC-style heuristic. 

We can also adapt a type of cross validation by masking a proportion of our entries, and solving for the low-rank estimate as a tensor completion problem. This is not typically used for standard CP decomposition, since the tensor completion problem is solved differently than the original decomposition (often with imputation (\cite{tomasi2005}), or with weights (\cite{Acar_2011_incomplete}) like we do in Section~\ref{sec:mean}), but likelihood-based methods lend themselves to tensor completion immediately \citep{Hong_GCP_2020, hood2026nearuniversalmultiplicativeupdatesnonnegative}. Suppose that we observe $\mX$ partially, and consider an index tensor  $\mT$ such that $\mT_{\bi}=1$ if the entry $\mX_{\bi}$ is observed, and zero otherwise. Our new loss function is:

$$
\ell(\mB,\mS) = 
-\dfrac{1}{2} \sum_{\bj}\mT_{\bj}\left( - \log(\mS_{\bj}) +  \mS_{\bj}(\mX_{\bj} -\mB_{\bj})^2  \right), 
$$
which can be solved with nearly the same alternating procedure we use for the complete data (we provide details for the tensor completion implementation in Appendix~\ref{sec:appendix-TC}). So, we can assess the fit of our low-rank model by analyzing relative errors or log-likelihoods in the completion problem as we mask different entries in our data. 

The convergence criterion, initialization, and maximum iterations are also user specified parameters for Algorithms \ref{alg:heterocp-mean},~\ref{alg:heterocp-precision}, and~\ref{alg:heterocp}. Since CP decompositions are not unique, we recommend using stationarity of the recovered tensor or of the log-likelihood, rather than of the factors, as convergence criterion. That being said, the algorithm will typically converge quickly, so for all experiments in this paper we simply use a small number of iterations for $i_{\mB}, i_{\mS},$ and $i_{\text{tot}}$. For 
$i_{\text{in}}$, one should use an especially small number of iterations because it is more beneficial to move on to the next factor with an improved estimate than it is to solve for the current factor perfectly. In our implementation we used $i_{\mB} = 5, ~ i_{\mS}= 5, ~ i_{\text{in}}= 5, $ and $i_{\text{tot}} = 20$. 

For initialization, we mentioned in Section~\ref{sec:heterocp} that we use standard CP decomposition to initialize the mean factors, and then run Algorithm~\ref{alg:heterocp-precision} with a larger number of iterations to initialize the precision factors (we used $50$). For this first run of Algorithm~\ref{alg:heterocp-precision}, we initialize the entries of each factor matrix randomly from $\operatorname{Gamma}(1, 1)$. 

Since ALS methods can occasionally get stuck in poor local optima, particularly in settings with missing data, we recommend using multiple seeds and selecting the result with the largest final log-likelihood. For our EEG experiment in Section \ref{sec:EEG}, we used the best of $5$ initializations for both HCP and CP. 

Finally, we note that the choice of homoskedastic modes or residual constraints in Sections~\ref{sec:partial-hetero} and~\ref{sec:constraint} are not viewed as parameters that need to be selected in practice the same way as the ranks. We would expect directions of homoskedastic variance to be chosen by the context or subject matter expertise prior to rank selection. If there is no contextual reason to consider partially heteroskedastic precision, we recommend simply defaulting to full heteroskedasticity, rather than trying to perform model selection over choices of homoskedastic modes and ranks simultaneously. For the lower bound on the residuals $\delta$, we implement as the default value $\delta = \sqrt{\epsilon_{\mathrm{mach}}}$, and only recommend increasing this if the degeneracy occurs. This is very easy to spot in practice and can be done simultaneously with rank selection; if a single entry of $\mS$ is many orders of magnitude larger than the others, then $\delta$ can be increased a few orders of magnitude. Like we see in Sections~\ref{sec:simulation} and \ref{sec:appendix-simulation}, the size of $\delta$ typically has little impact on the performance of HCP.

\section{Tensor completion} \label{sec:appendix-TC}
Suppose that we observe $\mX$ partially, and consider an index tensor  $\mT$ such that $\mT_{\bi}=1$ if the entry $\mX_{\bi}$ is observed, and zero otherwise. Our new loss function is:

$$
\ell(\mB,\mS \mid \mT) = 
-\dfrac{1}{2} \sum_{\bj}\mT_{\bj}\left( - \log(\mS_{\bj}) +  \mS_{\bj}(\mX_{\bj} -\mB_{\bj})^2  \right)
$$

For fixed $\mS$, we simply use Algorithm \ref{alg:heterocp-mean} with weights $\mT*\mS$. Note, though, that we should no longer expect the efficiency improvement from low-rank $\mS$, since, unless $\mT$ is appropriately structured, $\mT*\mS$ will not be low-rank. 

For fixed $\mB$, we follow the same MM argument as in Section~\ref{sec:precision} to get the following mode-$d$ precision-factor update:
\begin{equation} \label{eq:tc-update}
    \widetilde{\bS_d}
    \leftarrow
    \widetilde{\bS_d}
    *
    \left(
        \left[
            \mT_{(d)} \oslash (\widetilde{\bS_d}\bOmega_{d})
        \right]
        \bOmega_{d}^\top
    \right)
    \oslash
    \left(
        (\mT_{(d)} * \mR_{(d)})
        \bOmega_{d}^\top
    \right).
\end{equation}

Rank-one precision, partially heteroskedastic precision, and constrained residuals are handled by following the steps described in Sections~\ref{sec:rank-one},~\ref{sec:partial-hetero}, and~\ref{sec:constraint} to the modified log-likelihood. 
In the partially heteroskedastic case, we note the denominator in Equation \eqref{eq:avg-res} is the number of observed entries, instead of the size of the homoskedastic modes, and both the mask and residual tensors should be aggregated over the homoskedastic modes. This distinction is important under missingness, since different retained entries may summarize different numbers of observed tensor entries. 
Also, the penalty in Theorem~\ref{thm:penalty} will be $\frac{\delta}{2}\| \mT *\mS\|_1$, since the perturbation should only be added to the observed entries of $\mR$. 
Since tensor completion is not the focus of this paper, though, we will not provide a full, detailed, Algorithm in the likeness of Algorithm~\ref{alg:heterocp-full} here. 

We note that the final Hadamard divisor in Equation~\ref{eq:tc-update} is far more likely to have an entry of $0$ than its counterpart in Equation~\ref{eq:MM-update}, making including a constraint as in Section~\ref{sec:constraint} even more necessary. This happens when all observed entries for a row of matricized squared residuals are $0$, becoming more likely when fewer entries are observed. 

Finally, recall in Theorem~\ref{thm:GPCP}, we require a positive entry in each row of $\bA$. In the fully observed case, $\bA = \bone_{N_d} \bone_{N_{-d}}^\top$, so this assumption is met automatically, but in this tensor completion case, an entire row of \(\mX_{(d)}\) could be unobserved, leaving $\bA = \mT_{(d)}$ with no nonzero entries. Disallowing such structured missingness is necessary, since if an entire row of \(\mX_{(d)}\) is unobserved, then the corresponding factor row is unidentifiable from the observed likelihood. 

We should also expect a rank condition, analogous to that in CP-APR \cite{chiandkolda12}, on the columns of \(\bOmega_d\) selected by the observed entries in each row of \(\mX_{(d)}\). A necessary (though not sufficient) condition for \cite{chiandkolda12} is at least \(R\) observed entries in each row of \(\mX_{(d)}\), for all \(d \in [D]\), which is already much more strict than our assumption in Theorem~\ref{thm:GPCP}. We leave a full analysis of the tensor-completion properties of HCP to future work.

\section{Basic Illustrative Examples} \label{sec:appendix-basic}

\subsection{Algorithm~\ref{alg:heterocp-mean} Basic Example} \label{sec:appendix-basic-ex-mean}
For a basic warm-up example of Algorithm~\ref{alg:heterocp-mean}, we consider low-rank $\mB = [\![ \blambda_{\mB};\bB_1,\bB_2,\bB_3]\!]$ with $\bB_1 \in \mathbb{R}^{10 \times 3}, \bB_2 \in \mathbb{R}^{8 \times 3}, \bB_3 \in \mathbb{R}^{12 \times 3}$, each generated entrywise randomly from the standard normal distribution. 

We observe $\mX_{\bj} \sim \mN( \mB_{\bj} , \mS_{\bj}^{-1})$, where $\mS$ is constructed to decrease linearly along the first mode and vary quadratically along the third mode. In particular,
\[
\mS_{ijk}
=
\frac{1}{
\left(0.5\,i\right)
\left(0.25\bigl((k-3)^2+1\bigr)\right)
},
\qquad
i \in [N_1],\;
j \in [N_2],\;
k \in [N_3].
\]

Figure~\ref{fig:mean_res} plots the first horizontal slice and the first lateral slice for $\mX$, the true mean $\mB$, and the recovered $\widehat{\mB}$ from Algorithm~\ref{alg:heterocp-mean}, which we see is a good estimate. The fourth column provides the recovered mean from CP decomposition, for comparison, which we see performs far worse. While this is not a fair comparison -- with the known variance tensor, HCP simply has access to more information than CP -- we get confirmation that Algorithm~\ref{alg:heterocp-mean} can incorporate heteroskedasticity knowledge to improve our estimate of $\mB$. 

\begin{figure}[H]
    \centering
    \includegraphics[width=0.7\textwidth]{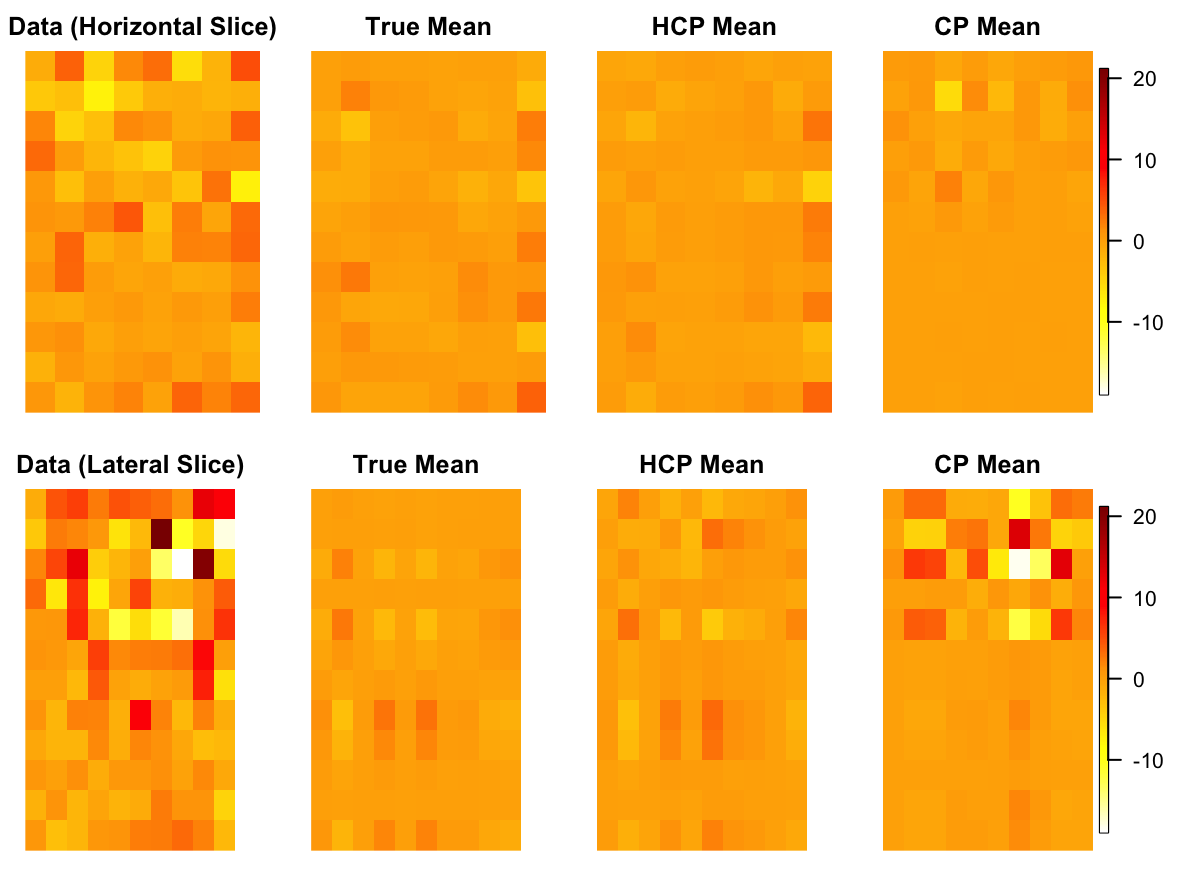} 
    \caption{Results for a basic 3 mode tensor ($10 \times 8 \times 12$) with \( R = 3 \) and known $\mathcal{S}$. Columns display observed data, the true mean, and the estimated mean from HCP and CP, for the first horizontal and lateral slices. 
    }
    \label{fig:mean_res}
\end{figure}

\subsection{Algorithm~\ref{alg:heterocp-precision} Basic Example} \label{sec:appendix-basic-ex-var}
Here, we reverse the example in Appendix~\ref{sec:appendix-basic-ex-mean}, setting $\mS = [\![\blambda_{\mS};\bS_1,\bS_2,\bS_3]\!],$ with $\bS_1 \in \mathbb{R}_+^{10\times 3}$, $\bS_2 \in \mathbb{R}_+^{8\times 3}$, and $\bS_3 \in \mathbb{R}_+^{12\times 3}$. We generate the factor entries by randomly selecting approximately $1/R$ of the entries to be generated from $\operatorname{Unif}(1,2)$, and generate the remaining entries from $\operatorname{Unif}(0.1,0.2)$. This more complicated structure is necessary since the precision factor matrices are positive \citep{chiandkolda12}; if we used a single distribution for all entries, their common positive mean can reinforce an approximately rank-one structure in $\mS$. We skip construction of $\mB$ and simply set it to zero,  since only the residuals enter the precision update, anyway. 

Figure~\ref{fig:var_res} plots the first horizontal and lateral slices for the observed residuals $\mR$, the true variance $\mS^{*-1}$, and the recovered $\widehat{\mS}^{*-1}$ from Algorithm~\ref{alg:heterocp-precision}. Like the subroutine for estimating the mean, given the precision, we see that this subroutine for estimating the precision, given the mean, has strong recovery. Without any competitors, though, we can't provide a comparison as we do in Appendix~\ref{sec:appendix-basic-ex-mean}. 

\begin{figure}[H]
    \centering
    \includegraphics[width=0.55\textwidth]{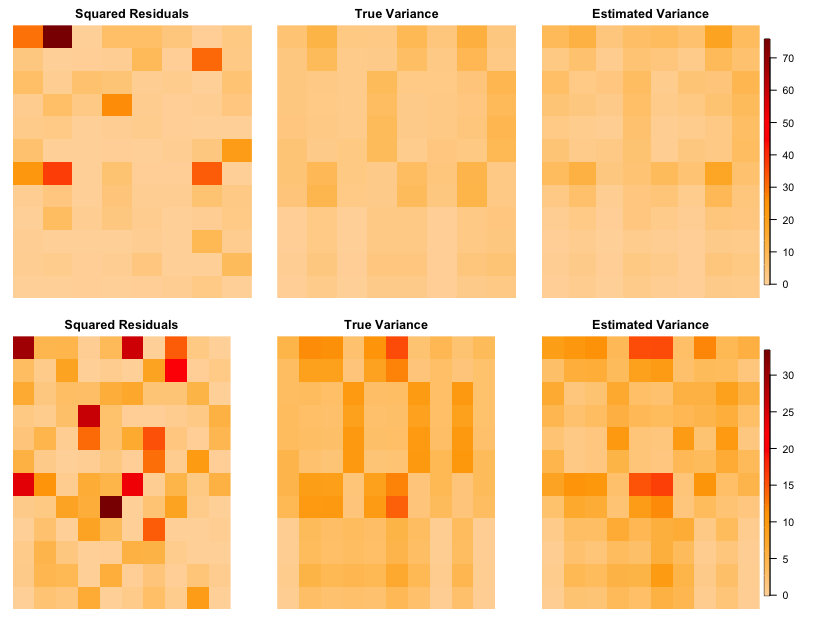} 
    \vspace{-1em}
    \caption{Results for a 3 mode tensor ($10 \times 8 \times 12$) with \( R = 3 \) and known $\mathcal{B} = 0$. Columns display the squared residuals, true variance, and estimated variance, for the first horizontal and lateral slices. 
    }
    \label{fig:var_res}
\end{figure}

\subsection{Algorithm~\ref{alg:heterocp} Basic Example} \label{sec:appendix-basic-ex-both}

Now, with strategies in place to alleviate the problems from the unbounded likelihood, we move on to an example that combines the examples of Appendix~\ref{sec:appendix-basic-ex-mean} and~\ref{sec:appendix-basic-ex-var}. Specifically, both $\mB$ and $\mS$ are rank $3$, and the factor matrices are generated as in Sections~\ref{sec:appendix-basic-ex-mean} and~\ref{sec:appendix-basic-ex-var}, respectively. We also assume $\mS$ is constant along the third mode, and leverage the techniques in Section~\ref{sec:partial-hetero}. We set the constraint (as in Section~\ref{sec:constraint}) to $\delta=0.001$. 

We plot the first frontal slices of the observed data, true mean and variance, and estimated mean and variance, in Figure~\ref{fig:both_res}, where we see that the recovery is strong for both parameters. We also provide the estimate for the mean from CP decomposition, for comparison, which does appear to be visually worse than HCP. 

\begin{figure}[H]
    \centering
    \includegraphics[width=0.7\textwidth]{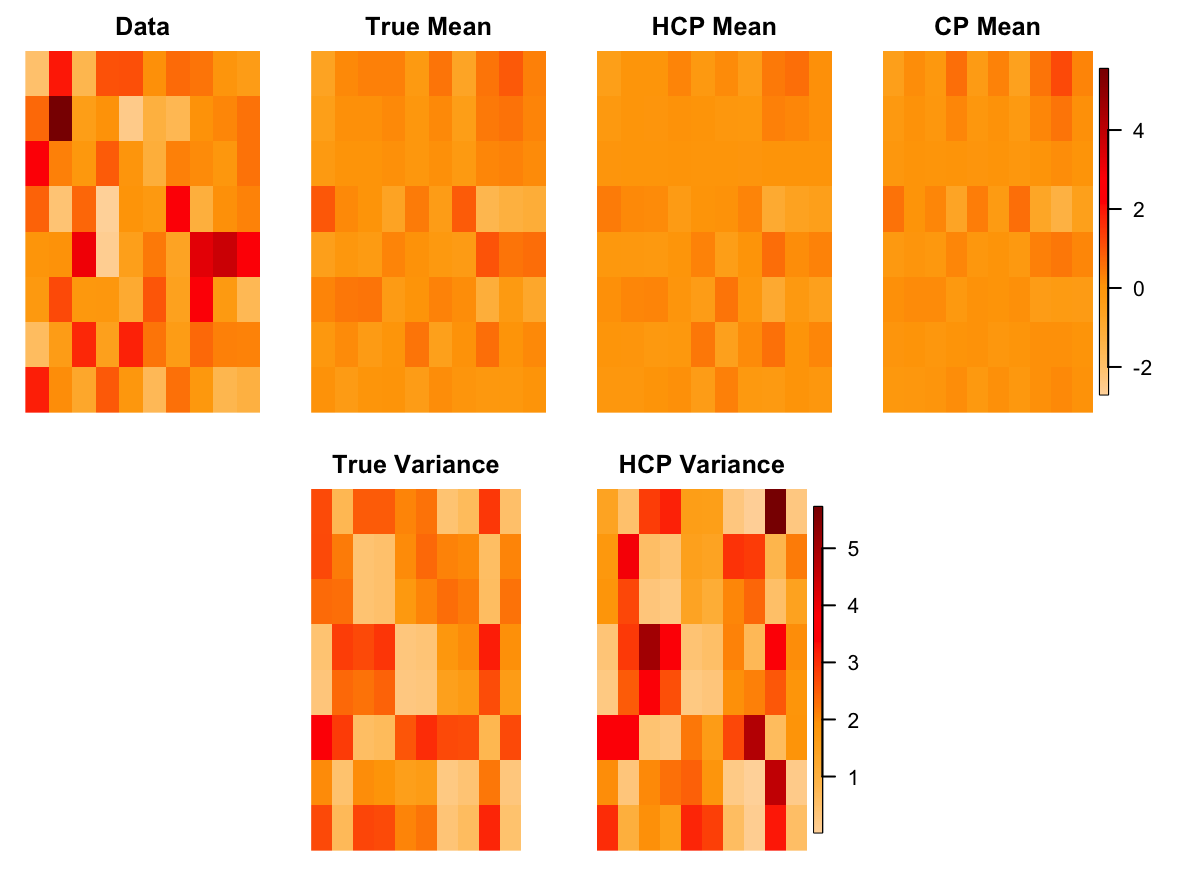} 
    \vspace{-1em}
    \caption{Results for a 3 mode tensor ($10 \times 8 \times 12$) with \( R = 3 \). 
    For the first frontal slice we show the data, true mean tensor, and estimated mean tensors from HCP and CP in the first row, and true vs learned variance tensors in the second row.}
        \label{fig:both_res}
    \end{figure}

\section{Expanded Parameter Recovery Benchmark} \label{sec:appendix-simulation}

In this section, we provide more details regarding the simulation study in Section~\ref{sec:simulation}, as well as provide additional results for a fully heteroskedastic precision tensor. In addition, we analyze the effect of rank misspecification. 

We will scale up the basic example from Appendix~\ref{sec:appendix-basic-ex-both} to varying tensor sizes and assess the recovery over many simulations. Since the true low-rank tensors $\mB$ and $\mS$ are known, 
we use $R_{\mB} = R_{\mS} = 3$ and measure performance with Relative Errors $\frac{\|\mB - \widehat{\mB} \|_{F}}{\|\mB\|_{F}}$ and $\frac{\|\mS - \widehat{\mS} \|_{F}}{\|\mS\|_{F}}$. $\mX$ will be size $m_{\text{mult}} \cdot \left(10 \times 8 \times 12 \right)$, with $m_{\text{mult}} \in \{1.5, 2, 3, 6\}$. We will present here the results both for partial and fully heteroskedastic $\mS$ (so $Q = 1$ and $Q = 0$, respectively, in the notation of Section~\ref{sec:partial-hetero}), repeating the results from Section~\ref{sec:simulation}, where only partially heteroskedastic $\mS$ is considered, for easier comparison. For our constraint we use $\delta = 0.001$ and $\delta = \sqrt{\epsilon_{\mathrm{mach}}}$, labeled HCP and HCP0, respectively. We will compare our method against CP decomposition for the mean recovery. For the recovery of the precision, there are no other competitors. 

Figure~\ref{fig:sim-results} provides our results as boxplots over 50 simulations. As we can see from Figure~\ref{fig:mean-sim}, our methods, regardless of the constraint or the exploitation of partial heteroskedasticity, outperform CP decomposition. For the recovery of the precision tensor, we have no competitors to compare against, but, as expected, our relative errors decrease with the size of the tensor. In addition, there is almost no difference in the mean or precision recovery between HCP and HCP0, particularly for the larger tensor sizes. This motivates our choice of $\delta = \sqrt{\text{Machine-Precision}}$ as the default, since in many cases this small constraint is sufficient. Finally, as expected, estimation when $\mS$ is fully heteroskedastic is a more difficult task. 

\begin{figure}[ht]
    \centering
    \begin{subfigure}[b]{0.48\textwidth}
        \centering
        \includegraphics[width=\textwidth]{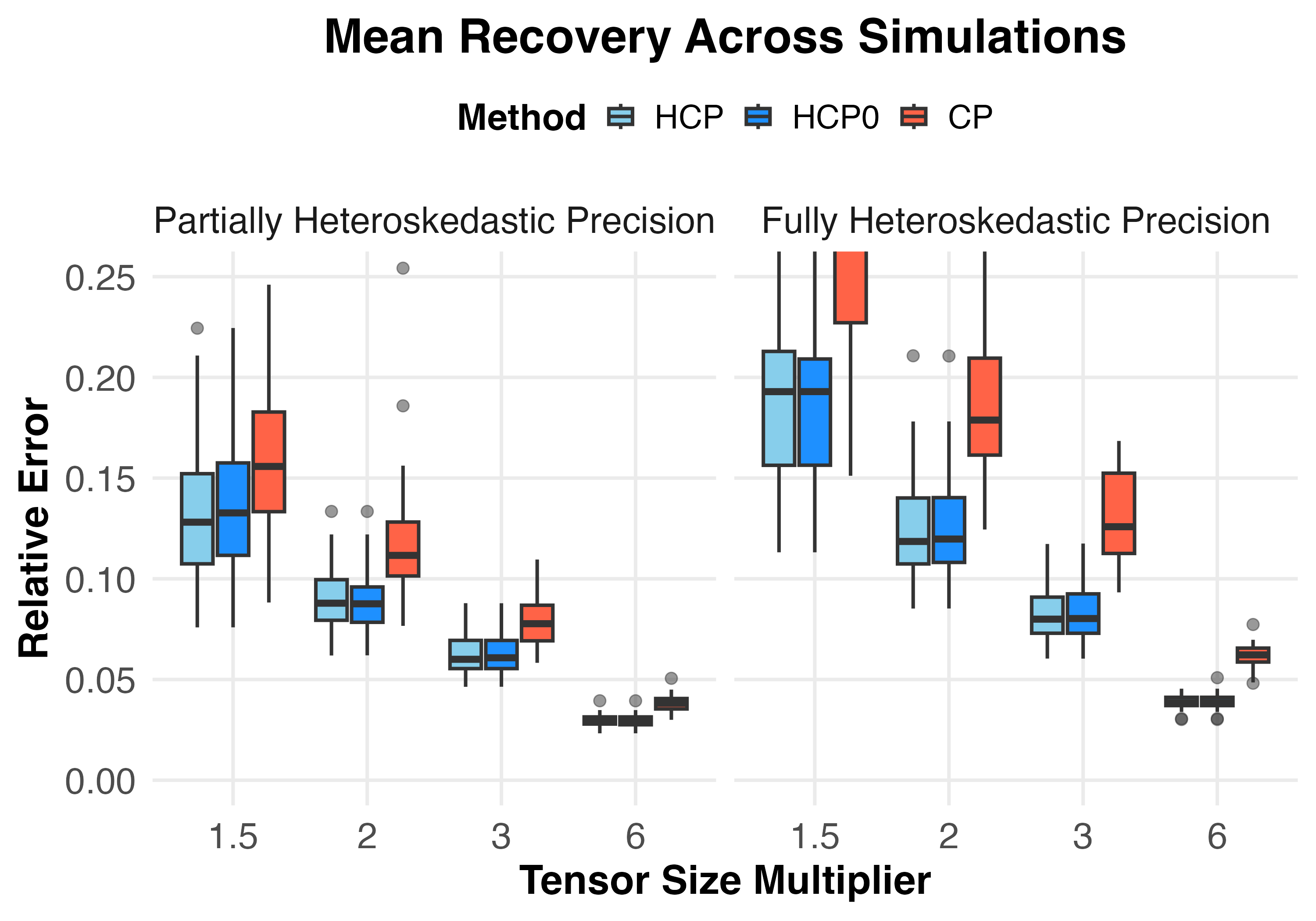}
        \caption{Relative error for mean recovery across simulations.}
        \label{fig:mean-sim}
    \end{subfigure}
    \hfill
    \begin{subfigure}[b]{0.48\textwidth}
        \centering
        \includegraphics[width=\textwidth]{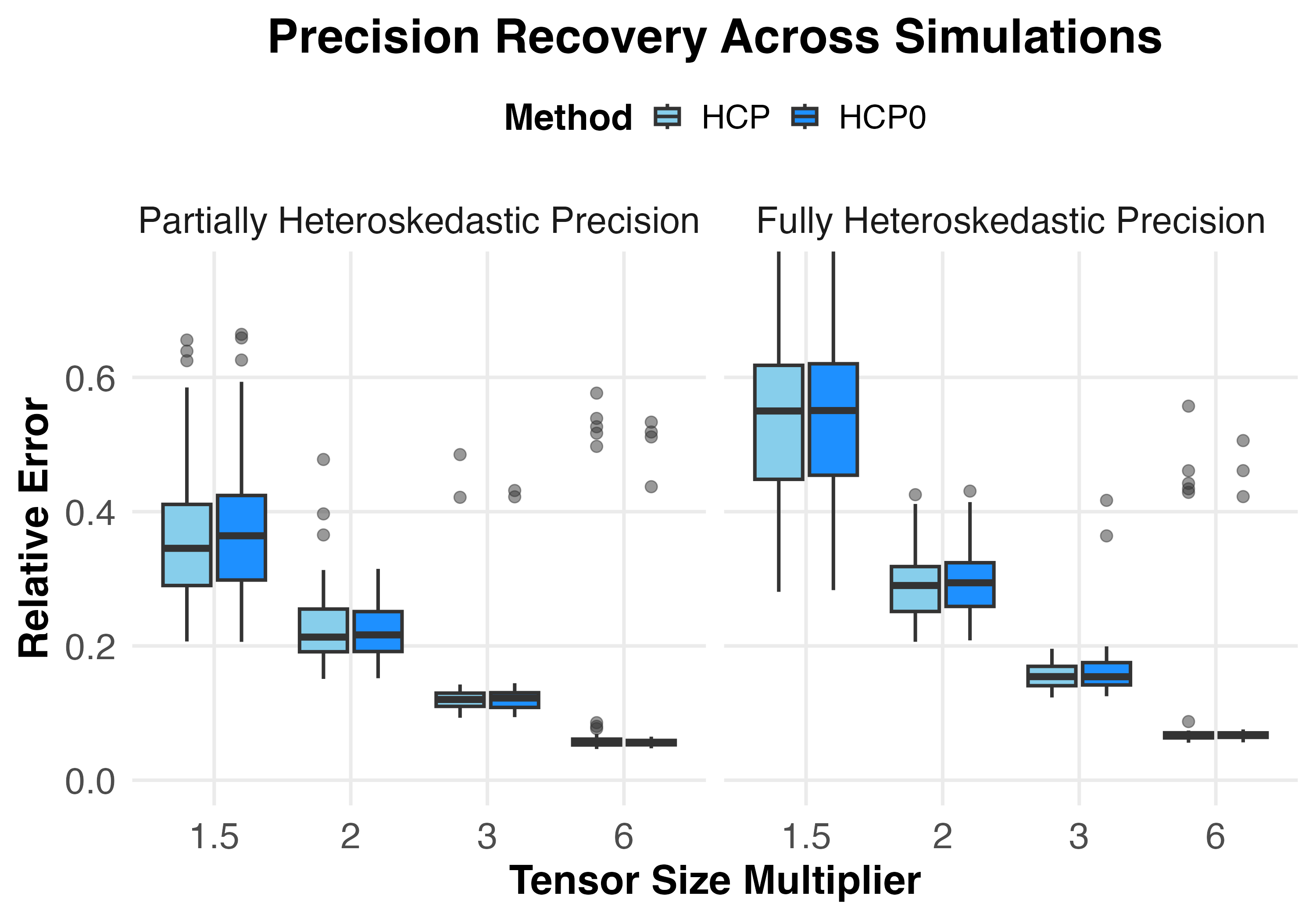}
        \caption{Relative error for precision recovery across simulations.}
        \label{fig:var-sim}
    \end{subfigure}
    \caption{Simulation results comparing recovery performance across tensor sizes and heteroskedasticity structures. The left panel shows relative error in recovery of $\mB$, while the right panel shows relative error in recovery of $\mS$. }
    \label{fig:sim-results}
\end{figure}

Next, we will intentionally fit models with ranks different from the truth (\(R=3\)) in order to assess whether performance degrades gracefully under rank misspecification. Rather than considering a full grid over mean and precision ranks, we vary one rank at a time while holding the other fixed at its true value. Specifically, we consider two rank-robustness experiments:
\[
R_{\mB} \in \{1,2,3,4,5\}, 
\qquad 
R_{\mS} = 3,
\]
and
\[
R_{\mB} = 3,
\qquad 
R_{\mS} \in \{1,2,3,4,5\}.
\]
This design isolates the effect of misspecifying the rank of the mean model from the effect of misspecifying the rank of the precision model. Ordinary CP decomposition is included as a baseline when studying mean recovery; for the varying mean-ranks, CP is fit with the corresponding mean rank. Performance is again measured using relative Frobenius error. 

Figures~\ref{fig:meanmean}--\ref{fig:varvar} summarize the results. Figure~\ref{fig:meanmean} shows mean recovery when the fitted mean rank is varied while the precision rank is fixed at its true value. Figure~\ref{fig:varmean} shows precision recovery under this same mean-rank variation, illustrating how misspecification of the mean model can also affect estimation of the heteroskedastic structure through the residuals.
Figure~\ref{fig:meanvar} shows mean recovery when the precision rank is varied while the mean rank is fixed at its true value, assessing whether mean estimation is sensitive to misspecification of the variance model. Finally, Figure~\ref{fig:varvar} shows precision recovery when the fitted precision rank is varied. The labels ``Partially Heteroskedastic'' and ``Fully Heteroskedastic'' indicate whether the precision tensor is constrained to be constant along one mode or allowed to vary along all modes. 

\begin{figure}[]
    \centering
    \includegraphics[width=0.9\textwidth]{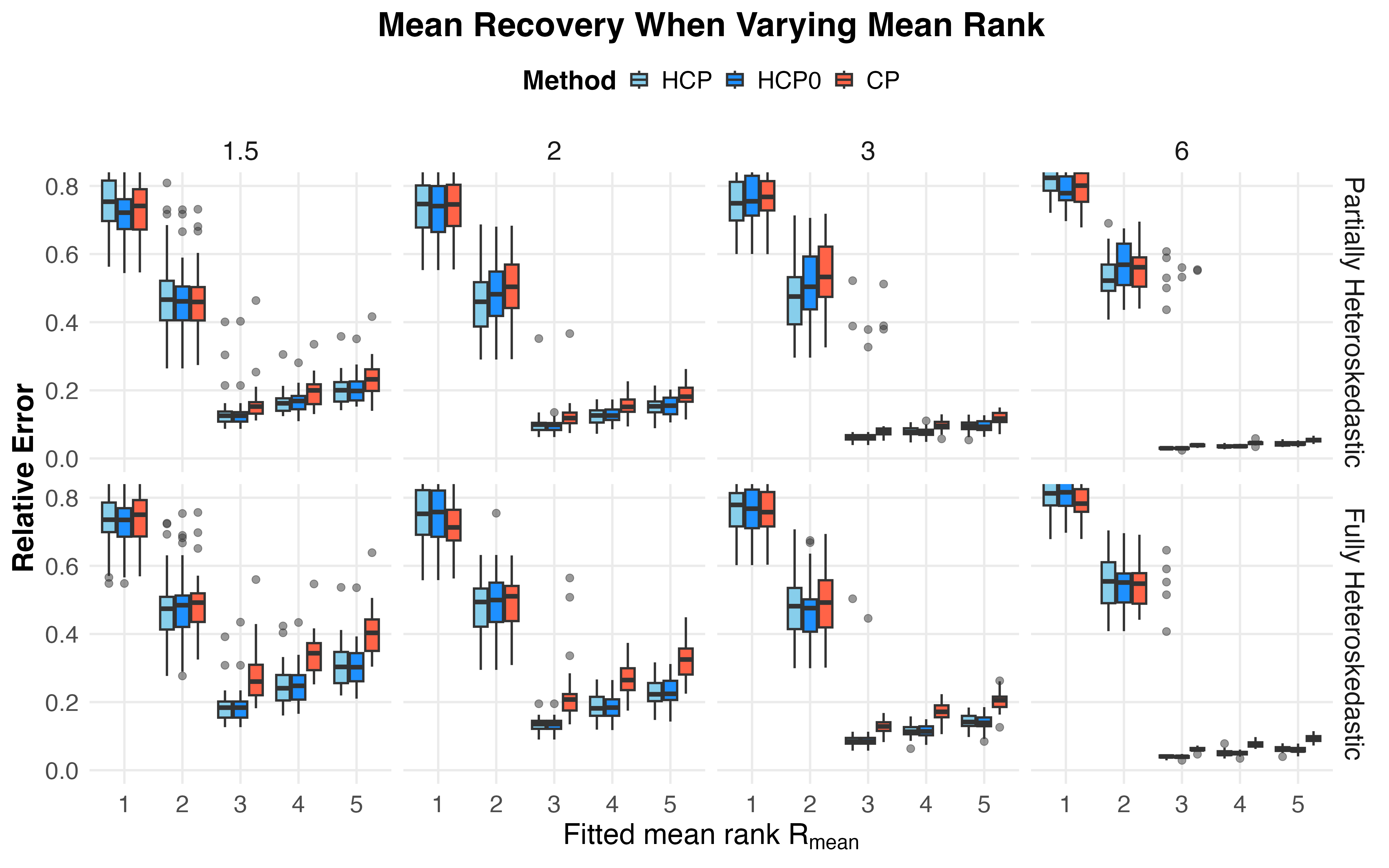}
    \caption{
    Mean recovery when varying the fitted mean rank \(R_{\mB}\), with the precision rank fixed at its true value \(R_{\mS}=3\). The true mean rank is \(R_{\mB}=3\). The worst performance for both methods occurs when the fitted mean rank is too small. 
    }
    \label{fig:meanmean}
\end{figure}

\begin{figure}[]
    \centering
    \includegraphics[width=0.95\textwidth]{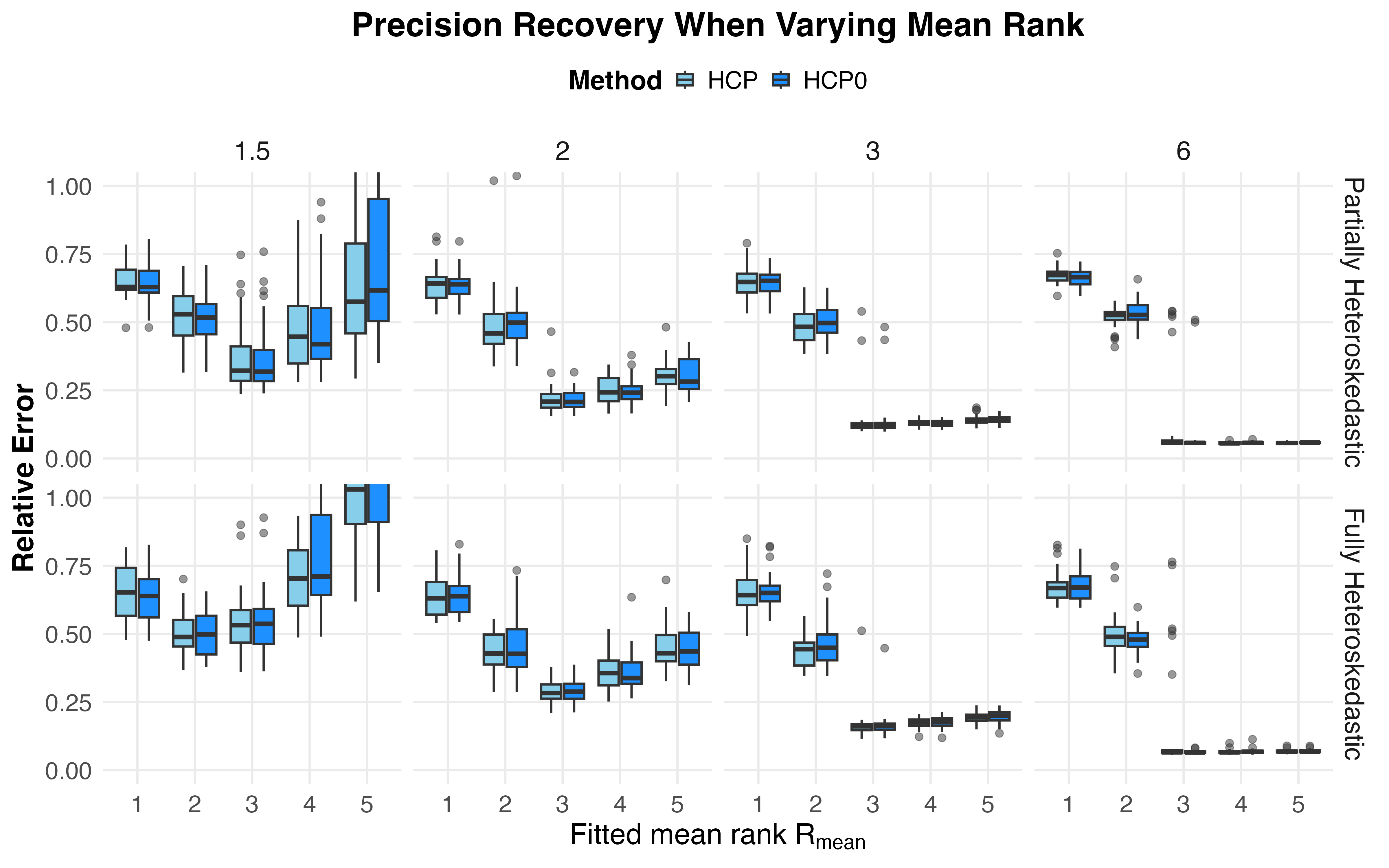}
    \caption{
    Precision recovery when varying the fitted mean rank \(R_{\mB}\), with the precision rank fixed at its true value \(R_{\mS}=3\). This plot shows how misspecification of the mean rank affects estimation of the precision tensor. 
    }
    \label{fig:varmean}
\end{figure}

\begin{figure}[]
    \centering
    \includegraphics[width=0.95\textwidth]{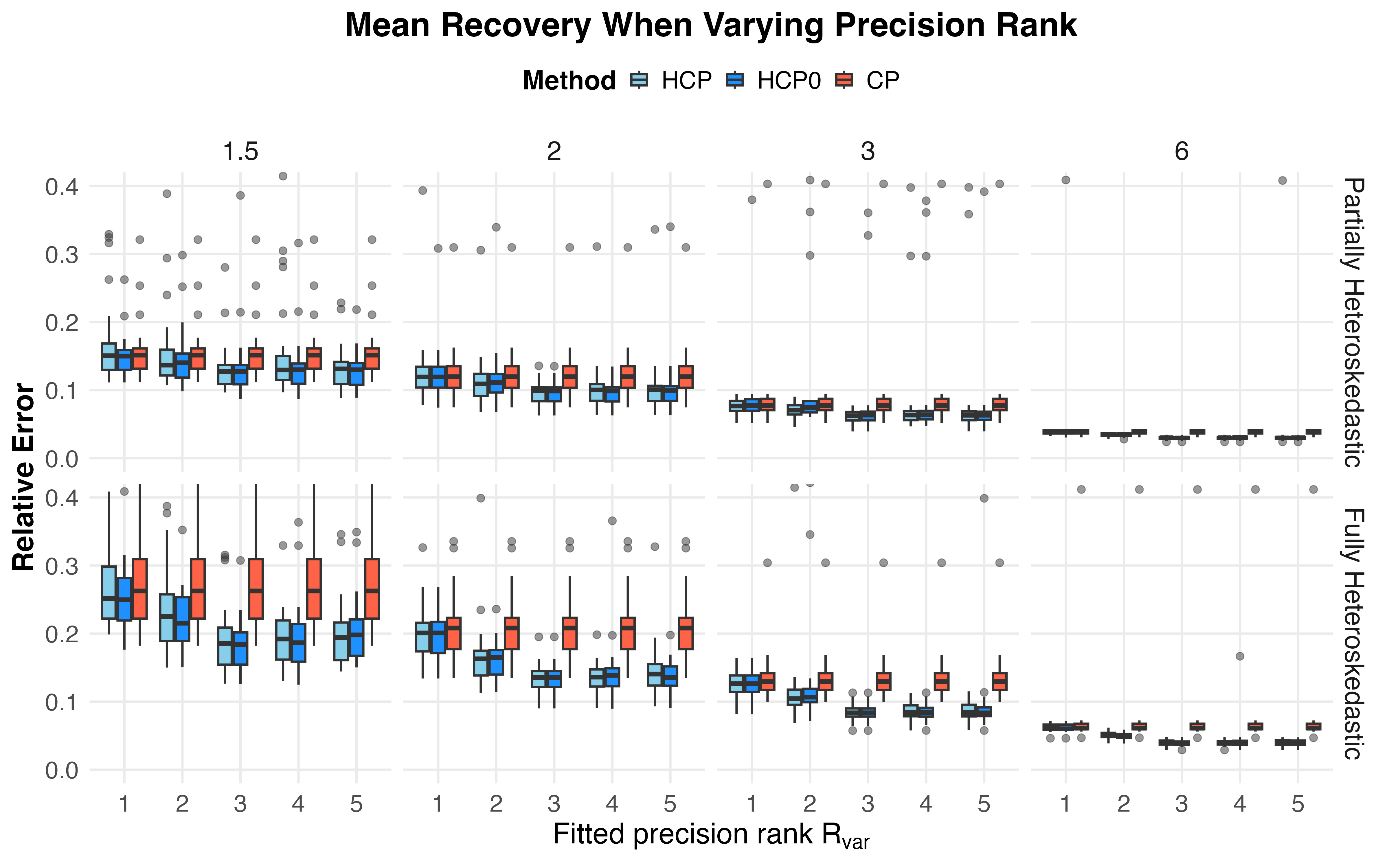}
    \caption{
    Mean recovery when varying the fitted precision rank \(R_{\mS}\), with the mean rank fixed at its true value \(R_{\mB}=3\). Importantly, even for the worst fitted precision ranks, HCP outperforms CP in mean recovery. 
    }
    \label{fig:meanvar}
\end{figure}

\begin{figure}[]
    \centering
    \includegraphics[width=0.95\textwidth]{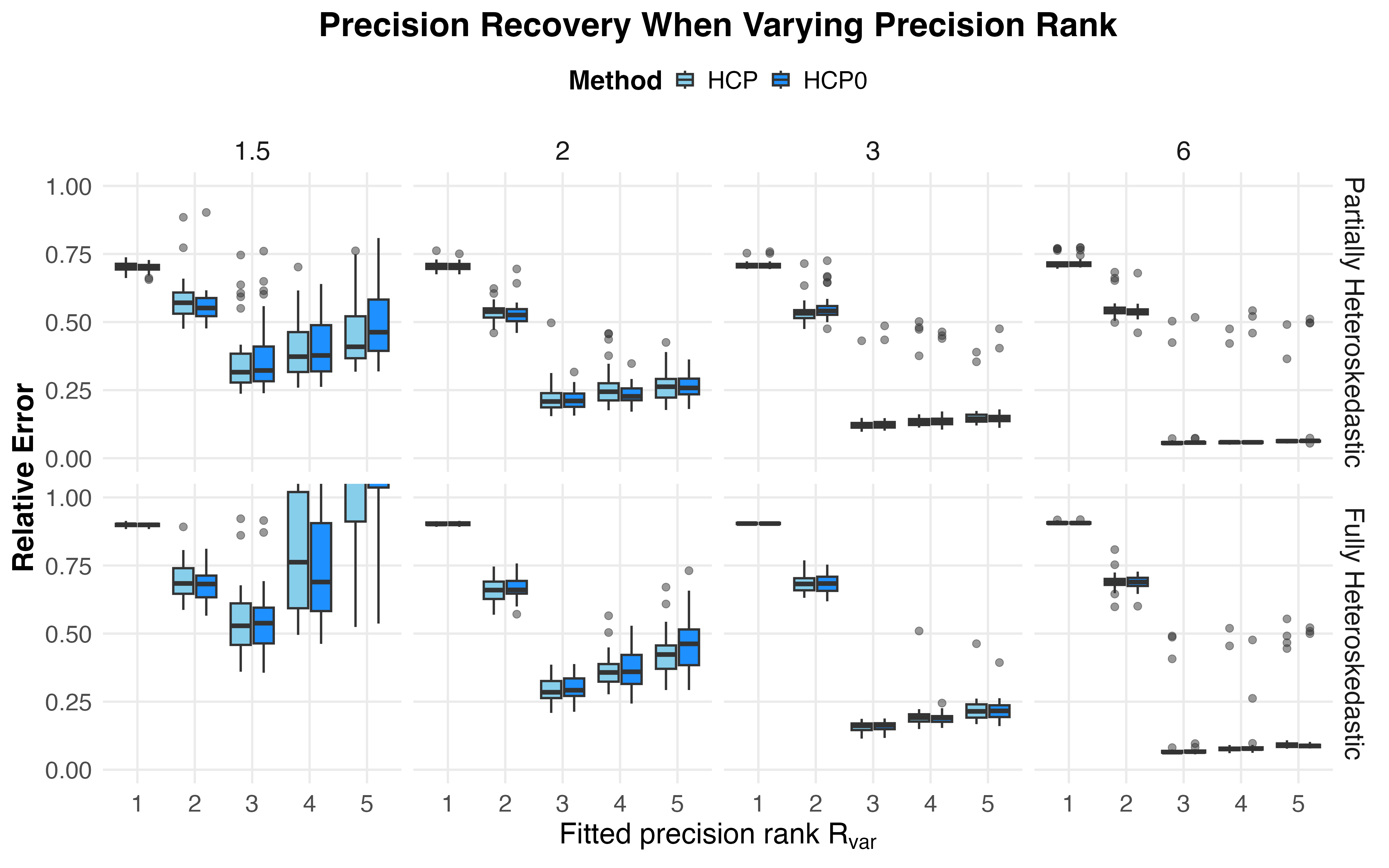}
    \caption{
    Precision recovery when varying the fitted precision rank \(R_{\mS}\), with the mean rank fixed at its true value \(R_{\mB}=3\). The true precision rank is \(R_{\mS}=3\). Unlike for the mean, if the tensor is small, the fitted rank being too large can be more detrimental than being too small. 
    }
    \label{fig:varvar}
\end{figure}

As we can see, it is generally more detrimental to underestimate the rank than it is to overestimate the rank, consistent with standard model selection discussions for CP \citep{core2003, koldaandbader09}. 
For HCP, since an improvement in the fitted log-likelihood does not imply a lower Frobenius error to the true mean, this problem is even more severe, and in Figure~\ref{fig:meanmean} we see that underestimating $R_{\mB}$ can sometimes lead to CP  outperforming HCP. Furthermore, for small tensors, the precision recovery does heavily deteriorate when the fitted rank is too large, indicating a likely concern with overfitting in these cases.

We conclude by drawing attention to Figure~\ref{fig:meanvar}, the mean recovery where \(R_{\mS}\) varies. For practitioners using HCP, it would be problematic if a poor choice of \(R_{\mS}\) could deteriorate HCP's performance below that of CP, but even in the most extreme case of \(R_{\mS} = 1\), the estimate of $\widehat{\mB}$ from HCP outperforms that of CP decomposition. Combined with the results from Section \ref{sec:dist-study}, we have evidence that, in the presence of heteroskedastic noise, the mean estimate benefits from incorporating noise information, even if characterization of the noise structure is incorrect.

\section{Additional Details for Experiments}
\subsection{Experiment Summary} \label{sec:appendix-experiment-table}

In this section we provide a summary of the experimental designs for all of the experiments performed in the main manuscript and the appendix sections. 

\renewcommand{\arraystretch}{2}
\begin{table}[H]
\centering
\caption{Summary of HCP experimental design choices. \(\mathcal{H}\subseteq [D]\) denotes the homoskedastic modes in the precision tensor $\mS$, and $\mathcal{H} = \varnothing$ refers to fully heteroskedastic precision. The constraint column gives the perturbation value \(\delta\) used in the precision update.}
\label{tab:hcp-experiment-design-summary}
\scriptsize
\begin{tabular}{lllllll}
\toprule
& Experiment
& Dimension 
& \makecell[l]{Ranks \\ HCP} 
& \makecell[l]{Ranks \\ Competitor(s)} 
& \makecell[l]{Homoskedastic \\ modes $(\mathcal{H})$}  
& Constraint $(\delta)$  \\

\midrule
& \makecell[l]{Basic Demo 1 \\ (\ref{sec:appendix-basic-ex-mean})} 
& \(10\times 8\times 12\)
& \(R_{\mB}=3\), ~~~~~ --
& \(R=3\) 
& ~~ --
& ~~ --\\

& \makecell[l]{Basic Demo 2 \\ (\ref{sec:appendix-basic-ex-var})} 
& \(10\times 8\times 12\)
& ~~~~ -- ~~~, \(R_{\mS}=3\),
& ~~ --
& \(\mathcal{H}=\varnothing\)
& $\delta =0$ \\

& \makecell[l]{Basic Demo 3 \\ (\ref{sec:appendix-basic-ex-both})} 
& \(10\times 8\times 12\)
& \(R_{\mB}=3\),~ \(R_{\mS}=3\)
& \(R=3\) 
& \(\mathcal{H}=\{3\}\)
& $\delta = 0.001$ \\

& \makecell[l]{Butterfly \\ (\ref{sec:butterfly})} 
& \(200\times 200\times 17\)
& \(R_{\mB}=16\), \(R_{\mS}=3\)
& \(R=16\) 
& \(\mathcal{H}=\{3\}\)
& \(\delta = \sqrt{\epsilon_{\mathrm{mach}}}\) \\

& \makecell[l]{Distributions \\ (\ref{sec:dist-study})} 
& \(40\times 32\times 48\)
& \(R_{\mB}=3\),~ \(R_{\mS}=3\)
& \(R=3\) 
& \(\mathcal{H}=\varnothing\)
& \(\delta = \sqrt{\epsilon_{\mathrm{mach}}}\) \\

& \makecell[l]{Recovery \\ (\ref{sec:simulation}, \ref{sec:appendix-simulation})} 
& \makecell[l]{\((10\times 8\times 12) * m_{\text{mult}}\)
\\ $m_{\text{mult}} \in \{ 1.5, 2, 3, 6\}$}
& \(R_{\mB}=3\),~ \(R_{\mS}=3\)
& \(R=3\) 
& \(\mathcal{H}\in\{\{3\},\varnothing\}\)
& \(\delta \in \{ 0.001, \sqrt{\epsilon_{\mathrm{mach}}} \} \) \\

& \makecell[l]{Varying-ranks \\ (\ref{sec:appendix-simulation})} 
& \makecell[l]{\((10\times 8\times 12) * m_{\text{mult}}\)
\\ $m_{\text{mult}} \in \{ 1.5, 2, 3, 6\}$}
& \makecell[l]{\(R_{\mB} \in [5]\), \(R_{\mS}=3\)
\\ \(R_{\mB} = 3\),~ \(R_{\mS}\in [5]\)}
& \makecell[l]{\(R \in [5] \)\\ \(R=3\)} 
& \(\mathcal{H}\in\{\{3\},\varnothing\}\)
& \(\delta \in \{ 0.001, \sqrt{\epsilon_{\mathrm{mach}}} \} \) \\

& \makecell[l]{EEG \\ (\ref{sec:EEG},\ref{sec:appendix-EEG-preprocessing})} 
& \(64\times 6099\times 45\)
& \(R_{\mB}=3\),~ \(R_{\mS}=1\)
& \(R \in \{3, 4\}\) 
& \(\mathcal{H}=\{3\}\)
& \(\delta = \sqrt{\epsilon_{\mathrm{mach}}}\) \\

\bottomrule
\end{tabular}
\end{table}

\renewcommand{\arraystretch}{1}

\subsection{Generalized CP Decomposition Implementation} \label{sec:appendix-gcp}
The experiments in Section~\ref{sec:dist-study} compare HCP with Generalized CP (GCP)  \citep{Hong_GCP_2020}. We use this section to provide details on our GCP implementation. 

For a data tensor \(\mX\), GCP estimates a low-rank tensor \(\widehat{\mB}=[\![\blambda;\bB_1,\ldots,\bB_D]\!]\) by solving
\[
    \min_{\widehat{\mB}}
    \sum_{\bj} f(\mX_{\bj},\widehat{\mB}_{\bj}),
\]
where \(f\) is chosen to match a distributional model. The estimated low-rank tensor corresponds to the mean or natural parameter of the chosen loss. 

All GCP competitors were fit using \texttt{pyttb} \citep{TTB_Software} with an L-BFGS-B optimizer. We matched ranks to HCP's $R_{\mB}$. 

The Gamma GCP used the built-in \texttt{pyttb} Gamma objective. For Poisson GCP, we did not use the default Poisson objective wrapper, which expects integer count data. Instead, we passed the Poisson loss and gradient handles directly. For the inverse-Gaussian GCP, we used the entrywise loss and gradient
\[
    f(x,m)
    =
    \frac{(x-m)^2}{2xm^2}, \qquad 
    \frac{\partial f}{\partial m}
    =
    \frac{m-x}{m^3}.
\]
The Gamma and inverse-Gaussian objectives require positive observations. Since Gaussian and Poisson generating families do not guarantee strictly positive observations in principle, we verified this condition empirically for every retained synthetic tensor used in Section~\ref{sec:dist-study}. No shifting, clipping, truncation, resampling, or offset was applied in the experiments of Section~\ref{sec:dist-study}.

\subsection{EEG Preprocessing and Rank-4 Results}
\label{sec:appendix-EEG-preprocessing}

This appendix provides additional details on the preprocessing steps used to construct the EEG tensor analyzed in Section~\ref{sec:EEG}, as well as some additional results.

We use the publicly available EEG Motor Movement/Imagery Dataset from PhysioNet \citep{Goldberger2000PhysioNet, PhysioNetEEGMMI}, collected using the BCI2000 system \citep{Schalk2004BCI2000}. The full dataset contains 64-channel EEG recordings from 109 subjects, with each subject completing 14 experimental runs. The EEG signals were recorded from 64 electrodes placed according to the international 10-10 system, excluding electrodes Nz, F9, F10, FT9, FT10, A1, A2, TP9, TP10, P9, and P10. 
Thus, throughout Section~\ref{sec:EEG}, a ``channel'' refers to one of these electrode locations on the scalp. We generated the scalp maps in Figure~\ref{fig:eeg-brainfactors} from MNE-Python \citep{Gramfort2014MNE}.

In our analysis, we use a single subject performing task \(1\), where a target appears on either the left or the right side of the screen, and the subject opens and closes the corresponding fist until the target disappears. This corresponds to runs 3, 7, and 11 in the PhysioNet protocol. The event annotations \(T1\) and \(T2\) correspond to left- and right-fist movement onsets, respectively, while \(T0\) corresponds to rest. We only retain the task-related \(T1\) and \(T2\) epochs. Each epoch begins at the annotated movement onset and extends from \(0\) to \(4\) seconds relative to that onset. Since the PhysioNet EEG Motor Movement/Imagery data are sampled at \(160\) Hz, each un-decimated epoch contains \(641\) time samples, including both endpoints.

The raw EEG recordings are loaded from the EDF files, the EEGBCI channel names are standardized, and a standard 10-05 montage is assigned using MNE-Python \citep{Gramfort2014MNE}. We retain only EEG channels and apply a \(1\)--\(40\) Hz bandpass filter prior to the time-frequency transform.

For each epoch and channel, we compute a complex Morlet wavelet
time-frequency transform. Let \(x_{i,k}(t)\) denote the EEG signal from
channel \(i\) during epoch \(k\). For each frequency \(f\), the Morlet
transform produces complex coefficients of the form $W_{i,f,t,k} \in \mathbb{C},$
indexed by channel \(i\), frequency \(f\), time \(t\), and epoch \(k\).
In our implementation, we use the frequency grid $f \in \{4,6,8,\ldots,40\}\text{ Hz},$  giving \(19\) frequencies in total. The number of Morlet cycles is chosen to depend on frequency according to $n_{\mathrm{cycles}}(f) = \frac{f}{2}.$ Thus, for example, the \(4\) Hz wavelet uses \(2\) cycles, the \(10\) Hz wavelet uses \(5\) cycles, and the \(40\) Hz wavelet uses \(20\) cycles. The time-frequency coefficients are decimated by a factor of \(2\), so the effective time grid has spacing \(2/160 = 0.0125\) seconds. Over the \(0\)--\(4\) second epoch window, this yields \(321\) time samples. We
convert the complex-valued coefficients to real-valued features by taking
their magnitudes, leaving, for each channel
and epoch, $19 \times 321 = 6099$
real-valued coefficients.

The resulting array initially has dimensions
\[
    \text{epoch} \times \text{channel} \times \text{frequency} \times
    \text{time}.
\]
For each channel and epoch, the frequency-by-time matrix of wavelet
magnitudes is vectorized into a single time-frequency feature vector.
Stacking these vectors yields a third-order tensor
\[
    \mX \in \mathbb{R}^{64 \times 6099 \times 45},
\]
whose modes correspond to channel, vectorized time-frequency coefficient,
and measurement/epoch. The \(45\) measurements correspond to the
\(T1\) and \(T2\) task epochs extracted from the three task-1 runs.

Finally, as in \cite{Acar_2011_incomplete}, we center the tensor across
the channel mode. Specifically, for each time-frequency coefficient \(j\)
and epoch \(k\), we subtract the mean over channels:
\[
    \widetilde{X}_{i,j,k}
    =
    X_{i,j,k}
    -
    \frac{1}{64}
    \sum_{i'=1}^{64} X_{i',j,k}, 
\]
and rescale the tensor entries by \(10^6\), so that the final tensor used in the factorization is expressed on a microvolt scale. 

Our preprocessing follows the structure used by \cite{Acar_2011_incomplete}: complex wavelet coefficients are converted to real values using magnitudes, the data are arranged as a channel-by-time-frequency-by-measurement tensor, and the tensor is centered across the channel mode. However, their EEG example uses a proprietary dataset and reports a time-frequency vector of length \(4392\), whereas our PhysioNet preprocessing uses \(19\) frequencies and \(321\) time samples, yielding \(6099\) time-frequency coefficients.

Finally, we present an additional result, where we fit CP decomposition with $R=4$. Since HCP requires extra parameters to fit the precision tensor, the total number of parameters with $R_{\mB}=3$ and $R_{\mS}=1$ is nearly\footnote{Since $\mS$ is partially heteroskedastic, it has $64 + 6099$ parameters, while an additional rank for CP adds $64 + 6099 + 45$ parameters} the same number of parameters as CP with $R=4$, making comparisons to both CP with $R=3$ and $R=4$ logical. The recovered factors, alongside the original factors with $R=3$ from Figure~\ref{fig:eeg-brainfactors}, are presented in Figure~\ref{fig:eeg-brainfactors-extra}. As we can see, the additional component does not allow CP to capture the localized channel signal that HCP captures in the third component. Instead, CP uses this extra component to further model the low-frequency pattern that is most likely variance, rather than true signal. 

\begin{figure}[]
    \centering
    \includegraphics[width=\textwidth]{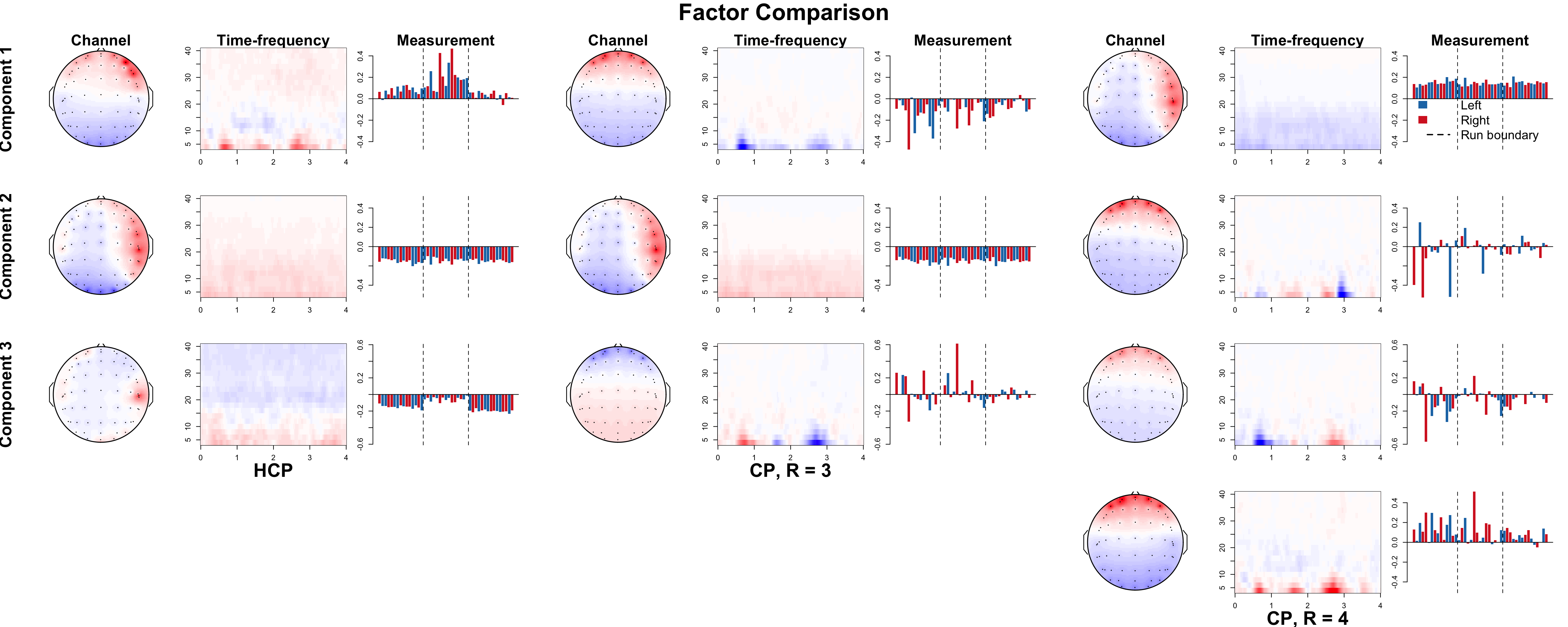}
    \caption{Visualization of the factor matrices recovered from the EEG tensor by HCP and CP with rows corresponding to components. Within each method, columns show the channel/scalp loadings, time-frequency loadings, and measurement loadings. Dashed vertical lines in measurement columns indicate recording-run boundaries.}
        \label{fig:eeg-brainfactors-extra}
\end{figure}

\subsection{Computing Infrastructure \& Software}
\label{sec:appendix-computing}

\noindent\textbf{Computing Infrastructure.}
The experiments in Sections \ref{sec:butterfly},  \ref{sec:EEG}, and \ref{sec:appendix-basic} were run on a 16-inch MacBook Pro with
    2.4 GHz 8-Core Intel Core i9 CPU, 32 GB of RAM, and no dedicated GPU.

The experiments in Sections \ref{sec:simulation}, \ref{sec:dist-study}, and \ref{sec:appendix-simulation} were run on a single node of an HPC cluster equipped with dual 14-core Intel Xeon processors (28 physical cores per node). We used 1 node, 1 task, and 4 CPUs per task.

\medskip

\noindent\textbf{Software.}
EEG preprocessing was performed in Python using
MNE-Python \citep{Gramfort2014MNE}, including the EEGBCI data-loading utilities
and the Morlet time-frequency transform. Tensor computations were performed in
R using \texttt{rTensor} \citep{rTensor}. Experiments involving generalized CP
tensor factorization used the \texttt{pyttb} implementation \citep{Hong_GCP_2020, TTB_Software}.

\end{document}